\documentclass[11pt]{article}

\usepackage[T1]{fontenc}
\usepackage[utf8]{inputenc}
\usepackage{lmodern}
\usepackage{cite}
\usepackage[a4paper,margin=1in]{geometry}
\usepackage{amsmath,amssymb,amsthm,mathtools}
\usepackage{graphicx}
\usepackage{booktabs}
\usepackage{enumitem}
\usepackage[hidelinks]{hyperref}
\usepackage{xcolor}
\hypersetup{
  colorlinks=true,
  citecolor=blue!60,
  linkcolor=black,
  urlcolor=blue!60
}
\usepackage[nameinlink,noabbrev]{cleveref}
\usepackage{microtype}
\usepackage{authblk}

\numberwithin{equation}{section}
\setlist[enumerate,1]{label=(\roman*)}

\newtheorem{thm}{Theorem}
\newtheorem{prop}{Proposition}

\newtheorem{cor}{Corollary}
\newtheorem{definition}{Definition}
\newtheorem{Lemma}{Lemma}

\newtheorem{property}{Property}

\newcommand{\ack}[1]{\section*{Acknowledgements}#1}

\begin{document}

\title{Novel Inconsistency Results for Partial Information Decomposition}

\author[1]{Philip Hendrik Matthias}
\author[1]{Abdullah Makkeh}
\author[1]{Michael Wibral}
\author[1]{Aaron J. Gutknecht\thanks{\texttt{agutkne@uni-goettingen.de}}}

\affil[1]{Campus Institute for Dynamics of Biological Networks,
Georg-August University Göttingen, Germany}

\date{}

\maketitle

\begin{abstract}
We prove that three cornerstone properties of classical information theory—non-negativity of information, the chain rule, and invariance under invertible transformations—become mutually incompatible when extended to Partial Information Decomposition (PID), the framework for decomposing the information that multiple sources carry about a target into redundant, unique, and synergistic contributions. This is, to our knowledge, the first PID incompatibility result involving only direct extensions of standard information-theoretic principles, without relying on any PID-specific assumptions. It reveals that PID necessarily operates in a different conceptual realm from classical information theory, where not all familiar properties can be preserved. Additionally, we re-examine the incompatibility result of Rauh et al. involving the so-called identity property and show that it holds under weaker assumptions than originally stated. Our proofs leverage the recently developed mereological framework for PID that makes the structural origins of these incompatibilities transparent and allows the results to be formulated for any PID-inducing concept, not just redundant information. We discuss the implications of these findings and examine which of the incompatible properties might be relaxed.
\end{abstract}

\section{Introduction}
Partial Information Decomposition (PID) aims to disentangle how information about a target variable is distributed across multiple sources, addressing the fundamental question of "who knows what?" about the target. This question emerges across virtually every scientific domain: from understanding how brain regions jointly encode stimuli, to determining which features in a dataset contribute to a prediction, to assessing how information is distributed in quantum systems.  The ubiquity of this question explains the breadth of PID's applications, spanning neuroscience \cite{wibral2017partial, varley2023partial, luppi2022synergistic}, machine learning \cite{shwartz2024compress, makkeh2025general, wollstadt2023rigorous, liang2023quantifying, schneider2025neuron, dewan2024diffusion},  quantum information \cite{vanenk2023quantum} and  algorithmic fairness \cite{dutta2023review,dutta2021fairness, Hamman2023}. 

Since the seminal work of Williams and Beer \cite{williams2010nonnegative}, researchers have proposed numerous measures for quantifying how information decomposes into redundant, unique, and synergistic components  (e.g.,~\cite{Harder2013, bertschinger2014quantifying, ince2017measuring, finn2018pointwise, rosas2020operational, makkeh2021introducing, Kolchinsky2022, gomes2024measure, Griffith2014, lyu2024explicit}). A central challenge in this endeavor has been identifying which properties such measures should satisfy. The literature has seen the proposal of many intuitively appealing properties (such as non-negativity of the components, certain monotonicity conditions, chain rules, etc.), yet no consensus has emerged on which combination of properties is essential \cite{liardi2026mathematical, rauh2023continuity, rauh2014reconsidering, Harder2013, finn2018pointwise}.

Adding to this challenge, it was discovered that certain property combinations are fundamentally incompatible. Bertschinger et al.~\cite{Bertschinger2013_shared_info}, Rauh et al.~\cite{rauh2014reconsidering}, and Finn \& Lizier \cite{finn2018pointwise} proved \emph{inconsistency results} showing that some properties cannot be satisfied simultaneously, no matter how a measure is constructed. At first glance, this may seem to make matters worse. However, such findings serve an important purpose: much like the impossibility of squaring the circle in classical geometry, they rule out unattainable goals and allow researchers to focus their efforts on property combinations that are actually achievable.

Consider a researcher who encounters a PID measure satisfying a highly valued property (perhaps non-negativity of all components) but exhibiting an undesirable feature (perhaps failure to satisfy a chain rule). A natural question arises: can the measure be modified to preserve the desirable property while eliminating the problematic one? If these properties are provably inconsistent, the answer is immediate: such a modification is impossible, and one must either accept the tradeoff or pursue an entirely different approach. Impossibility theorems have played similar roles in other fields: for instance, Arrow's theorem in social choice theory \cite{Arrow1951} proved that no ranked voting system can simultaneously satisfy a set of seemingly reasonable conditions, while recent work in algorithmic fairness \cite{algo_fairness_tradeoff} has demonstrated that different statistical notions of fairness are mutually exclusive. 

The inconsistency results established in this paper address some of the most fundamental desiderata for a PID. Our main theorem reveals that three properties—each a direct extension of core principles of classical information theory—cannot all be satisfied simultaneously: the \textit{target chain rule}, \textit{non-negativity}, and \textit{invariance under invertible transformations}. This means that the incompatibility arises not from controversial PID-specific assumptions like the \textit{identity property} \cite{Harder2013, ince2017measuring,james2018unique, finn2018pointwise}, but from straightforward extensions of well-established information-theoretic principles. This demonstrates that in extending information theory to the PID setting, we enter a conceptually different realm where not all familiar properties can be preserved. In fact, only weak forms of two of these properties are needed: non-negativity can be replaced by the requirements that a source without information about the target shares none of it and that two sources with complete information share all of it, and invariance under invertible transformations by invariance under swapping the components of the target.

Beyond establishing this central impossibility, we revisit the inconsistency theorem of Rauh et al. \cite{rauh2014reconsidering}, which showed that the "Williams and Beer axioms" together with non-negativity and the identity property cannot be jointly satisfied. We strengthen and clarify this result in four ways: we show that the identity property is needed only for independent sources, that the source monotonicity condition used in the original proof is not required, that the obstruction persists for every fixed $n\geq 3$, and that the argument relies on transformation invariance, an assumption that was used implicitly but not stated explicitly. Finally, we show that the inconsistency results are not tied to formulating the target chain rule and the identity property specifically for redundant information. Instead, these properties admit equivalent formulations for alternative information concepts that can be used to construct a PID, and the same incompatibilities arise under these reformulations.

\section{Background}

We first set out the basic notation that will be used throughout: We consider a collection of jointly distributed discrete random variables $\mathbf{S} = (S_1, \ldots, S_n)$ and $T$, where the $S_i$ are referred to as the sources, and $T$ is termed the target. We denote subsets of source variables by their indices. For example, the subset $\{S_1, S_2\}$ is represented simply as $\{1, 2\}$. Generic subsets of source indices are denoted by bold lowercase letters, such as $\mathbf{a}$ and $\mathbf{b}$, where $\mathbf{a}, \mathbf{b} \subseteq [n] := \{1, \ldots, n\}$. The mutual information provided by a specific subset $\mathbf{a}$ of the source variables about the target is denoted by $I(\mathbf{a}; T) := I\left( \{ S_i \}_{i \in \mathbf{a}}; T \right)$. The powerset of $[n]$ is written as $\mathcal{P}([n])$. \textit{Antichains} of the partial order $(\mathcal{P}([n]), \subseteq)$, i.e. sets $\{\mathbf{a}_1, \ldots,\mathbf{a}_m\}$ such that no $\mathbf{a}_i$ is a superset of $\mathbf{a}_j$ for $i\neq j$, are denoted by greek letters $\alpha, \beta$. The set of antichains excluding $\{\}$ and $\{\{\}\}$ is denoted by $\mathcal{A}_n$.

\subsection{Partial Information Decomposition} \label{sec:background:pid}
We adopt the formal definition of PID introduced in \cite{gutknecht2021bits, gutknecht2025babel}, to which we refer readers for additional details. To motivate the definition, consider first the case $n=2$. The four \textit{information atoms}---redundancy, unique information from $S_1$, unique information from $S_2$, and synergy---are characterized by their distinctive \textit{ parthood relationships} to the mutual information terms $I(S_1;T)$, $I(S_2;T)$, and $I(S_1,S_2;T)$ as shown in the following table:

\begin{table}[ht]
\centering
\begin{tabular}{|c|c|c|c|c|} 
\hline
$\emptyset$ & $\{1\}$ & $\{2\}$ & $\{1,2\}$ & Interpretation\\ 
\hline
0 & 1 & 1 & 1 & redundancy \\ 
0 & 1 & 0 & 1 & unique $S_1$ \\ 
0 & 0 & 1 & 1 & unique $S_2$ \\ 
0 & 0 & 0 & 1 & synergy \\ 
\hline
\end{tabular}
\caption{Parthood distributions for $n=2$ sources showing the four information atoms and their interpretations.}
\label{tab:parthood}
\end{table}

Each row represents a possible configuration of parthood relationships: a 1 in column ${\mathbf{a} \subseteq [n]}$ means the corresponding atom is part of $I(\mathbf{a};T)$, while 0 means it is not. For instance, the synergy is only part of the joint mutual information $I(S_1,S_2;T)$. This observation suggests a natural generalization to arbitrary $n$: enumerate all possible such configurations (corresponding to a row in the table) and associate with each one an information atom quantifying precisely the information standing in those parthood relationships.

To formalize this, we introduce the concept of \textit{parthood distributions}, Boolean functions that encode these configurations:

\begin{definition}[Parthood Distribution]
A \textit{parthood distribution} in the context of $n$ source variables is a function $f:\mathcal{P}([n]) \rightarrow \{0,1\}$ such that
\begin{enumerate}[label=B\arabic*]
	\item \label{axiom:parthood1} $f(\emptyset) = 0$ (no information in the empty set)
	\item \label{axiom:parthood2} $f(\{1,\ldots,n\}) = 1$ (all information is in the full set)
	\item \label{axiom:parthood3} $(\mathbf{a} \subseteq \mathbf{b} \text{ and } f(\mathbf{a})=1) \implies f(\mathbf{b})=1$ (monotonicity)
\end{enumerate}
\end{definition}
A parthood distribution describes a row in a table analogous to Table \ref{tab:parthood} above. The set of all parthood distributions for $n$ sources is denoted by $\mathcal{B}_n$. In a PID, there is exactly one information atom $\Pi(f)$ for each parthood distribution $f \in \mathcal{B}_n$~\cite{gutknecht2021bits}. The atom $\Pi(f)$ quantifies precisely the information that is part of $I(\mathbf{a};T)$ for all $\mathbf{a}$ with $f(\mathbf{a})=1$, and not part of $I(\mathbf{a};T)$ for all $\mathbf{a}$ with $f(\mathbf{a})=0$.

This leads naturally to the following definition:

\begin{definition}[Partial Information Decomposition]\label{def:PID}
A partial information decomposition is a family of real-valued functionals $\{\Pi(f;T)\}_{f \in \mathcal{B}_n}$ of the joint distribution of $(\mathbf{S},T)$ satisfying the consistency equations
\begin{equation}\label{eq:consistency}
I(\mathbf{a};T) = \sum_{\substack{f \in \mathcal{B}_n \\ f(\mathbf{a})= 1 }} \Pi(f;T)
\quad \text{for all } \mathbf{a} \subseteq [n].
\end{equation}
\end{definition}

The consistency equations simply state that each mutual information term equals the sum of all atoms that are part of it.

Note that Definition~\ref{def:PID} encompasses the original Williams--Beer framework as a special case (see Proposition~2.8 in \cite{gutknecht2025babel}). Williams and Beer introduced the notion of a redundancy function based on certain axioms for redundant information, and then defined the PID atoms indirectly via a M\"obius inversion over the so-called \textit{redundancy lattice}. The key advantage of our alternative characterization is that it introduces the information atoms directly (the quantities of ultimate interest) rather than proceeding indirectly via a redundancy measure. In particular, this clarifies that redundancy plays no privileged role in PID theory: other notions such as union information or synergistic information can equally serve as PID-inducing concepts, and the relationships between these different approaches to the PID problem become transparent. This is discussed further in Section \ref{sec:generality}.

In addition to PID we also introduce the notion of a \textit{conditional PID} which will be important in the context of our main results:
\begin{definition}[Conditional Information Atoms] \label{def:cond_atoms}
Let $(\mathbf{S},T,Z)$ be jointly distributed random variables, with $Z$ an external variable or part of $\mathbf{S}$ or $T$. For any $f \in \mathcal{B}_n$, the conditional information atom is defined by

\begin{align}
\Pi(f;T \mid Z) = \sum_{z} p(z) \Pi^{p(\mathbf{S},T \mid Z=z)}(f;T).
\end{align}
where the superscript $p(\mathbf{S},T \mid Z=z)$ indicates that we are evaluating the same functional $\Pi(f;T)$ under the conditional distribution of $(\mathbf{S},T)$ given $Z=z$.

\end{definition}
This definition parallels the standard construction of conditional mutual information. In fact, in Fano's classical exposition of information theory \cite{Fano1961} he made it an axiom that however "information" be defined, "conditional information" should be the same functional evaluated under the conditional distribution. 

As a final remark we note that in the PID literature, it is common to index information atoms by antichains $\alpha \in \mathcal{A}_n$ rather than by parthood distributions $f \in \mathcal{B}_n$, writing $\Pi(\alpha;T)$ instead of $\Pi(f;T)$. These two notations are equivalent via a natural bijection~\cite{gutknecht2021bits}: each parthood distribution $f$ corresponds to the antichain $\alpha_f$ consisting of the \textit{minimal} sets $\mathbf{a}$ for which $f(\mathbf{a})=1$. Conversely, each antichain $\alpha$ corresponds to the parthood distribution $f_\alpha$ that assigns 1 to all supersets of elements in $\alpha$. Formally, 
\begin{equation} 
f \mapsto \alpha_f := \min_{\subseteq}\{\mathbf{a} \subseteq [n] \mid f(\mathbf{a})=1\} 
\end{equation}
with inverse 
\begin{equation} 
\alpha \mapsto f_\alpha(\mathbf{b}) := \begin{cases} 1 & \text{if } \exists \mathbf{a} \in \alpha: \mathbf{b} \supseteq \mathbf{a} \\ 0 & \text{otherwise} \end{cases} 
\end{equation} 
For instance, in Table~\ref{tab:parthood}, the redundancy atom corresponds to the antichain $\{\{1\},\{2\}\}$ (since $\{1\}$ and $\{2\}$ are the minimal sets with value 1), while the synergy atom corresponds to $\{\{1,2\}\}$. We will make use of the antichain notation when referring to \textit{specific} atoms in the context of a given number $n$ of source variables. Here it becomes cumbersome to write out the full parthood distribution with its $2^n$ entries. The parthood distribution notation by contrast is particularly useful for general definitions and derivations for arbitrary $n$.

\subsection{PID-inducing Concepts}\label{sec:background:pid_inducing}

The consistency equations \eqref{eq:consistency} do not uniquely determine a solution to the PID problem. This is because the number of atoms equals the $n$th Dedekind number minus two, whereas the number of consistency equations is only $2^n - 1$. To resolve this underdetermination, one introduces the notion of \textit{PID-inducing concepts} \cite{gutknecht2025babel}. A PID-inducing concept is a family of information measures that, once specified, makes it possible to uniquely solve for all information atoms. The most famous example of a PID-inducing concept is \textit{redundant information}, which can be formally introduced as follows.

\begin{definition}[Redundant Information]\label{def:redundancy}
A real-valued functional $I_{\cap}(\mathbf{a}_1,\ldots,\mathbf{a}_m;T)$ of the joint distribution of $(\mathbf{S},T)$ is called a measure of redundant information if there exists a PID $\{\Pi(f;T)\}_{f \in \mathcal{B}_n}$ such that
\begin{equation}\label{eq:C_red}
I_{\cap}(\mathbf{a}_1,\ldots,\mathbf{a}_m;T) = \sum\limits_{\substack{f \in \mathcal{B}_n \\ \mathcal{C}_\mathrm{red}(\mathbf{a}_1,\ldots \mathbf{a}_m;f)}} \Pi(f;T) \quad \text{ for all } \mathbf{a}_i \subseteq [n], m \in \mathbb{N}
\end{equation}
where $\mathcal{C}_\mathrm{red}$ is a logical condition on parthood distributions given by
\begin{equation}
\mathcal{C}_\mathrm{red}(\mathbf{a}_1,\ldots \mathbf{a}_m;f) \Leftrightarrow \forall i \in [m]: f(\mathbf{a}_i) = 1
\end{equation}
\end{definition}
The definition is motivated by the simple idea that the redundant information between source collections $\mathbf{a}_1,\ldots,\mathbf{a}_m$ should consist exactly of those atoms that are part of \textit{all} of the mutual information terms $I(\mathbf{a}_i;T)$. But these are all atoms $\Pi(f)$ such that $f(\mathbf{a}_i)=1$ for all $i$. This constraint on the parthood distributions is specified in the logical condition $\mathcal{C}_\mathrm{red}$.

The definition of redundancy immediately implies a range of basic properties.
\begin{prop}\label{prop:implied}
    Any measure of redundant information satisfies the following properties:
    \begin{enumerate}
    \item $I_\cap(\mathbf{a}_1,\ldots,\mathbf{a}_m;T) = I_\cap(\mathbf{a}_{\sigma(1)},\ldots,\mathbf{a}_{\sigma(m)};T)$ for any permutation $\sigma$ (symmetry)
    \item $I_\cap(\mathbf{a};T) = I(\mathbf{a};T)$ (self-redundancy)
    \item If $\mathbf{a}_i \supseteq \mathbf{a}_j$ for $i\neq j$, then $I_\cap(\mathbf{a}_1,\ldots, \mathbf{a}_m;T) = I_\cap(\mathbf{a}_1,\ldots, \mathbf{a}_{i-1}, \mathbf{a}_{i+1},\ldots,\mathbf{a}_m;T)$ (superset invariance)
\end{enumerate}
\end{prop}
\begin{proof}
See proof of Proposition 2.10 in \cite{gutknecht2025babel}.
\end{proof}
Note that these properties coincide almost exactly with the so-called \textit{Williams and Beer axioms} for redundancy functions \cite{williams2011information, williams2010nonnegative}. The only difference is that the \textit{superset invariance} property corresponds only to the equality case of their monotonicity axiom. Redundancy measures according to Definition~\ref{def:redundancy} are \textit{not} required to satisfy the  inequality condition  
\begin{equation}
I_\cap(\mathbf{a}_1,\ldots,\mathbf{a}_m,\mathbf{a}_{m+1};T) \;\leq\; I_\cap(\mathbf{a}_1,\ldots,\mathbf{a}_m;T),
\end{equation}
which demands that adding a source collection to the argument cannot increase redundancy. This makes our definition of redundant information (and PID), slightly broader than that of Williams and Beer. The broader formulation is useful, as it encompasses approaches in the literature that do not satisfy source monotonicity (e.g. \cite{ince2017measuring, makkeh2021introducing}). Including such approaches is sensible, since there are at least debatable reasons why monotonicity might fail (in particular the possibility of misinformation as discussed in Section \ref{sec:discussion}). In any case, since this paper is concerned with incompatibilities between PID properties, adopting a wider definition can only strengthen the results. 

We further introduce the \textit{redundancy lattice}, which encodes the nesting structure of redundancy terms and will be used in the proofs of our main results.

\begin{definition}[Redundancy Lattice]\label{def:redundancy_lattice}
The redundancy lattice is the partially ordered set $(\mathcal{A}_n,\preceq)$, where
\begin{equation}
\alpha \preceq \beta \quad \iff \quad \forall \mathbf{b} \in \beta \; \exists \mathbf{a} \in \alpha : \mathbf{a} \subseteq \mathbf{b}.
\end{equation}
\end{definition}
The redundancy lattice is shown for $n=3$ in Figure \ref{fig:redundancy}. Definition \ref{def:redundancy} of redundant information can be rewritten in terms of the redundancy lattice (see the proof of Proposition~2.11 in \cite{gutknecht2025babel}):
\begin{equation}\label{eq:red_lattice}
I_{\cap}(\alpha;T) = \sum\limits_{\beta \preceq \alpha} \Pi(\beta;T)
\end{equation}
This system of equations (one for each $\alpha\in\mathcal{A}_n$) determines the atoms uniquely: each equation contains $\Pi(\alpha;T)$ with coefficient one and otherwise only atoms $\Pi(\beta;T)$ with $\beta\prec\alpha$, so the atoms can be computed recursively from the bottom of the lattice upwards as $\Pi(\alpha;T)=I_\cap(\alpha;T)-\sum_{\beta\prec\alpha}\Pi(\beta;T)$ (in closed form, by M\"obius inversion~\cite{williams2010nonnegative}). Redundant information is therefore a PID-inducing concept.

\begin{figure}
    \centering
    \includegraphics{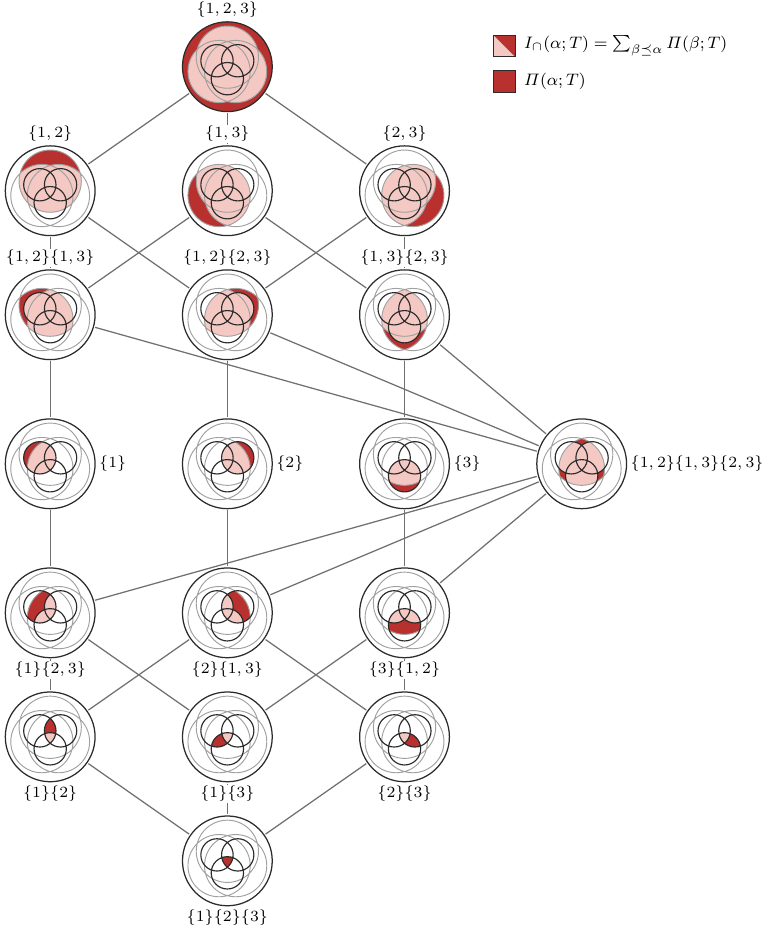}
    \caption{Redundancy lattice $(\mathcal{A}_3,\preceq)$ with an information diagram~\cite{williams2010nonnegative} at each node.
Antichains are written in shorthand, e.g.\ $\{1\}\{2,3\}$ for $\{\{1\},\{2,3\}\}$. At each node
$\alpha$, the entire shaded area (light and dark red together) represents the redundancy
$I_\cap(\alpha;T)=\sum_{\beta\preceq\alpha}\Pi(\beta;T)$; the dark red part highlights the atom
$\Pi(\alpha;T)$ itself. The shaded area of each node contains those of all nodes below it,
reflecting the nested structure of redundant information. By self-redundancy, the nodes $\{i\}$,
$\{i,j\}$ and $\{1,2,3\}$ show which circles represent $I(S_i;T)$, $I(S_i,S_j;T)$ and
$I(S_1,S_2,S_3;T)$. Region sizes are schematic and do not reflect amounts of information.}
    \label{fig:redundancy}
\end{figure}

Finally, we define \textit{conditional redundancy}, which will be needed to formulate the target chain rule for measures of redundant information. Following the same logic as the conditional information atoms in Definition~\ref{def:cond_atoms}, conditional redundancy is obtained by evaluating the same functional under the conditional distribution and averaging:
\begin{definition}[Conditional Redundancy]\label{def:cond_red}
Let $(\mathbf{S},T,Z)$ be jointly distributed random variables. Then
\begin{equation}
I_\cap(\alpha;T \mid Z) := \sum_{z} p(z)\, I_\cap^{p(\mathbf{S},T \mid Z=z)}(\alpha;T),
\end{equation}
where the superscript indicates evaluation under the conditional distribution of $(\mathbf{S},T)$ given $Z=z$.
\end{definition}

Although our main results are formulated in terms of redundant information, redundancy is just one of several PID-inducing concepts that can equivalently serve as the starting point for constructing a PID. In Section~\ref{sec:generality} we introduce the general framework of what we call $\mathcal{C}$-information and show that our inconsistency results are not specific to redundancy but hold under equivalent formulations for any PID-inducing concept.

\section{Candidate Properties}
In the PID framework, as presented in the previous section, one can distinguish between two kinds of properties. On the one hand, there are the minimal conditions that follow directly from the definitions: the consistency equations for the atoms as well as the implied properties of PID-inducing concepts such as redundant information, as described in Proposition \ref{prop:implied}. These are, at least from our perspective, constitutive background assumptions: without them, the atoms and the associated redundancy measures could not have the intended interpretation (see also discussion in Section \ref{sec:discussion} below).

In addition to this definitional core, the literature has proposed a number of further properties that one might wish to demand.  In this section, we present the properties that feature in our inconsistency results, providing a brief motivation for each. A more detailed discussion is provided in Section \ref{sec:discussion}.

\begin{property}[Local Positivity (\textbf{LP})]
All information atoms are non-negative:
\begin{equation}
\Pi(f;T) \geq 0 \quad \text{for all } f \in \mathcal{B}_n.
\end{equation}
\end{property}
Local positivity is a widely desired additional property for PID measures. It was highlighted by Williams and Beer as a central positive feature of their original $I_{min}$-based PID (although not stated as an explicit axiom), and it mirrors the non-negativity of mutual information in classical information theory. The assumption seems natural if information is understood as "uncertainty reduction" or as "dependence strength": in the extreme cases  (no reduction of uncertainty or complete independence) it may be zero, whereas the meaning of negative values would at least appear conceptually problematic.

\begin{property}[Re-encoding Invariance (\textbf{REI})]
Let $(S_1,\ldots,S_n,T)$ be jointly distributed random variables, and let 
\[
S_i' := \varphi_i(S_i), \quad T' := \psi(T),
\] 
where each $\varphi_i$ and $\psi$ is an invertible transformation on the supports of the respective variables.  
A PID satisfies re-encoding invariance if the values of the information atoms are unchanged under such transformations. Equivalently, the associated measure of redundant information is invariant:
\begin{equation}
I_\cap(\mathbf{a}_1,\ldots,\mathbf{a}_m;T) 
= I_\cap(\mathbf{a}_1',\ldots,\mathbf{a}_m';T'),
\end{equation}
where $\mathbf{a}_i'$ denotes the collection of re-encoded sources obtained by applying the bijections to all sources with indices in $\mathbf{a}_i$.
\end{property}
Re-encoding invariance requires that a PID’s values remain unchanged under invertible transformations of the sources or the target, ensuring that informational relationships do not depend on arbitrary choices of representation. In the discrete case, such transformations simply correspond to relabeling outcomes (for example, coding values $\{0,1,2,3\}$ as binary strings $\{00,01,10,11\}$). In this setting, \textbf{REI} holds automatically for any measure that depends on the joint distribution only through the probabilities of events, and not on how the outcomes are labelled or structured, which is perhaps the reason why it has not seen much explicit discussion in the PID literature. It is not automatic, however: a measure may exploit the structure of the outcome space and thereby violate \textbf{REI} (see Section~\ref{sec:discussion}). Mutual information itself is known to be transformation invariant both in the discrete and in the continuous case \cite{kraskov2004estimating}. 

\begin{property}[Target Chain Rule  for Redundant Information (\textbf{TCR})]
Suppose the target splits as $T=(T_1,T_2)$, then

\begin{equation}
I_\cap(\alpha;T) = I_\cap(\alpha;T_1) + I_\cap(\alpha;T_2 \mid T_1)
\quad \text{for all } \alpha \in \mathcal{A}_n.
\end{equation}

\end{property}
The target chain rule captures the idea that redundant information about a composite target can be decomposed into redundant information about its first component and the remaining redundant information about the second component once the first is known. This mirrors the familiar chain rule for mutual information. Note that \textbf{TCR} is directed as formulated here: $T_1$ is treated first and $T_2$ is conditioned on it. Although it is natural to require invariance under the permutation of target components, we intentionally keep this symmetry separate and capture it through re-encoding invariance (of which permutation invariance is a special case) so that its role remains explicit. 

\textbf{TCR} has been introduced in \cite{Bertschinger2013_shared_info}, and measures satisfying it have been proposed in \cite{finn2018pointwise} and \cite{makkeh2021introducing}. It will be proved below that the target chain rule for redundant information is equivalent to analogous target chain rules for any other PID-inducing concept as well as a target chain rule for the information atoms (see Proposition \ref{prop:equiv_tcr}). 

\begin{property}[Lattice Monotonicity of Redundant Information (\textbf{LM})]
Redundant information is monotone along the redundancy lattice:
\begin{equation}
\alpha \preceq \beta \;\;\Rightarrow\;\; 
I_\cap(\alpha;T) \leq I_\cap(\beta;T)
\end{equation}
\end{property}
This property reflects the natural idea that since the redundancy lattice describes the nested structure of redundancy terms (see Figure \ref{fig:redundancy} above), redundancies higher up in the lattice should be at least as large as those further below. It is worth noting, that lattice monotonicity is stronger than the following “source monotonicity” property:

\begin{property}[Source Monotonicity of Redundant Information (\textbf{SM})]
Redundant information is monotone under adding a source collection to its argument
    \begin{equation}
I_\cap(\mathbf{a}_1,\ldots,\mathbf{a}_m,\mathbf{a}_{m+1};T) \;\leq\; I_\cap(\mathbf{a}_1,\ldots,\mathbf{a}_m;T),
\end{equation}
\end{property}
This property is part of the monotonicity axiom of Williams and Beer \cite{williams2010nonnegative,williams2011information}. It concerns the effect of adding a source collection to the argument: Intuitively, any information shared by collections \( \mathbf{a}_1, \dots, \mathbf{a}_m, \mathbf{a}_{m+1} \) is also shared by collections \( \mathbf{a}_1, \dots, \mathbf{a}_m \), but not necessarily the other way around. Therefore, adding a collection should only decrease or maintain the amount of redundant information, never increase it.  This is in fact a special case of moving up the redundancy lattice since it is always the case that:
\begin{equation}
\{\mathbf{a}_1,\ldots,\mathbf{a}_m,\mathbf{a}_{m+1}\} \preceq \{\mathbf{a}_1,\ldots,\mathbf{a}_m\}
\end{equation}
However, moving up the lattice does not always correspond to simply adding a source collection; for instance, \( \{1\}\{2,3\} \) is below \( \{1,2\} \), even though we are not adding a new source in the latter case.  The justificatory rationale for \textbf{SM} and \textbf{LM} is, however, essentially the same: when one redundancy term is completely contained within another, the former should not be larger than the latter. Thus, if one accepts source monotonicity for this reason, one should equally accept lattice monotonicity.

\begin{property}[Identity Property (\textbf{ID})]\label{property:ID}
For two sources $S_i,S_j$,
\begin{equation}
I_\cap(\{i\},\{j\};(S_i,S_j)) = I(S_i;S_j).
\end{equation}

\end{property}
If the target is simply the concatenation of two sources, then the information they share about that target should be exactly the information they have about each other. This property was first introduced in \cite{Harder2013}, motivated by a critique of the original $I_{\min}$ measure of Williams and Beer (for more details see Section \ref{sec:discussion} below). There is also a weaker version of the identity property known as the \textit{independent identity property}  \cite{ince2017measuring} which only demands that if $S_1$ and $S_2$ are independent their redundant information about the target $T=(S_1,S_2)$ should be zero:

\begin{property}[Independent Identity Property (\textbf{IID})]
For two independent sources $S_i,S_j$,
\begin{equation}
I_\cap(\{i\},\{j\};(S_i,S_j)) = 0.
\end{equation}
\end{property}

\section{Inconsistency Results}
Here we present our main inconsistency results. First, we establish our main result: the joint assumptions of \textbf{LP}, \textbf{TCR}, and \textbf{REI}—each extending a hallmark property of classical information theory—are inconsistent. This result strengthens a related finding by Finn and Lizier \cite{finn2018pointwise}, who proved an analogous incompatibility but only under the additional assumption of the highly contentious \textbf{ID}. This makes our theorem, to the best of our knowledge, the only inconsistency result involving solely properties that extend standard information-theoretic principles.

Second, we strengthen and clarify the theorem of Rauh et al. \cite{rauh2014reconsidering} which states that no PID and associated redundancy function can satisfy the Williams and Beer axioms (these are the properties in Proposition \ref{prop:implied} plus \textbf{SM}), \textbf{LP} and \textbf{ID}. We show that \textbf{SM} is in fact not necessary for the statement, that \textbf{ID} can be weakened to \textbf{IID}, and that the inconsistency holds for every fixed $n\geq 3$. We also make explicit the assumption of \textbf{REI}, which was already used in the original proof by Rauh et al. but not stated as such.

After deriving bounds on pairwise redundancies under \textbf{LP} (Section~\ref{sec:pw_bounds}), we establish the incompatibility of \textbf{LP}, \textbf{TCR}, and \textbf{REI} directly for every fixed number of sources $n\geq 2$ by means of a two-source counterexample, and show that any two of the three properties are compatible, so that they form a minimal inconsistent set (Section~\ref{sec:main_result}). We then prove the incompatibility of \textbf{LP}, \textbf{IID}, and \textbf{REI} for every fixed $n\geq 3$ and show that these properties, too, form a minimal inconsistent set (Section~\ref{sec:id_result}). Finally, we discuss how the two results are related (Section~\ref{sec:relation}).

\subsection{\texorpdfstring{Pairwise redundancy bounds}{Pairwise redundancy bounds}}\label{sec:pw_bounds}

The following lemma derives bounds on pairwise redundancies that will be used in the proof of our main theorem and in relating the two inconsistency results.

\begin{Lemma}[Pairwise redundancy bounds]\label{lem:pw_bounds}
Assume \textbf{LP}. For distinct $i\neq j$,
\begin{align}
I(S_i;T) + I(S_j;T) - I(S_i,S_j;T)
&\mathrel{\le} I_\cap(\{i\},\{j\};T) \notag\\
&\le \min\{I(S_i;T),\,I(S_j;T)\}, \\
I(S_i;T\mid Z) + I(S_j;T\mid Z) - I(S_i,S_j;T\mid Z)
&\mathrel{\le} I_\cap(\{i\},\{j\};T\mid Z) \notag\\
&\le \min\{I(S_i;T\mid Z),\,I(S_j;T\mid Z)\}.
\end{align}
\end{Lemma}
\begin{proof}
\emph{Upper bounds.} By Definition~\ref{def:redundancy},
$I_\cap(\{i\},\{j\};T)$ is the sum of all atoms $\Pi(f;T)$ satisfying
$f(\{i\})=1$ and $f(\{j\})=1$. By the consistency
equations~\eqref{eq:consistency}, $I(S_i;T)$ is the sum of all atoms with
$f(\{i\})=1$. The former is a subset of the latter, and under \textbf{LP} all atoms are non-negative,
so $I_\cap(\{i\},\{j\};T) \le I(S_i;T)$. The same argument with $j$ in place of $i$ gives
$I_\cap(\{i\},\{j\};T) \le I(S_j;T)$, yielding the minimum bound.

\emph{Lower bounds.} By the consistency equations~\eqref{eq:consistency},
\begin{align}
I(S_i;T) + I(S_j;T) - I(S_i,S_j;T)
=
\sum_{f\in\mathcal{B}_n}
\bigl(
f(\{i\}) + f(\{j\}) - f(\{i,j\})
\bigr)\Pi(f;T).
\end{align}
Since $f$ is monotone, $f(\{i\})=1$ or $f(\{j\})=1$ implies
$f(\{i,j\})=1$. Hence the coefficient
\[
f(\{i\}) + f(\{j\}) - f(\{i,j\})
\]
is equal to $1$ if $f(\{i\})=f(\{j\})=1$ (these atoms form the redundancy between $S_i$ and $S_j$), equal to $0$ if exactly one of
$f(\{i\})$ and $f(\{j\})$ is $1$, and equal to either $0$ or $-1$ if both
$f(\{i\})$ and $f(\{j\})$ are $0$. Therefore,
\begin{align}
I(S_i;T) + I(S_j;T) - I(S_i,S_j;T)
&=
I_\cap(\{i\},\{j\};T)
-
\sum_{\substack{f\in\mathcal{B}_n\\
f(\{i\})=0,\, f(\{j\})=0,\, f(\{i,j\})=1}}
\Pi(f;T) \\
&\le I_\cap(\{i\},\{j\};T),
\end{align}
where the inequality uses \textbf{LP}.

The conditional bounds follow by applying the same reasoning to each conditional
distribution given $Z=z$ and averaging over $p(z)$.
\end{proof}

\subsection{Incompatibility of Local Positivity, Target Chain Rule, and Re-encoding Invariance}\label{sec:main_result}

We now present our main result demonstrating the fundamental incompatibility between local positivity, the target chain rule, and re-encoding invariance for information decomposition. The key idea of the following proof is to apply the
target chain rule to the \emph{same} composite target in two different orders, and to
observe that \textbf{REI} forces the two decompositions to agree, which under
\textbf{LP} they cannot.

\begin{thm}[Inconsistency of \textbf{LP}, \textbf{TCR}, and \textbf{REI}]\label{thm:lp_tcr_incompatible}
For every $n\geq 2$, there exists no PID $\{\Pi(f;T)\}_{f \in \mathcal{B}_n}$ and associated measure of redundant information $\{I_\cap(\alpha;T)\}_{\alpha \in \mathcal{A}_n}$ that simultaneously satisfy local positivity, the target chain rule and re-encoding invariance.
\end{thm}

\begin{proof}
Assume for contradiction that such a PID exists. Let $S_1,S_2 \sim \mathrm{Bern}(1/2)$
be independent, define
\begin{equation}
X := S_1 \oplus S_2,
\end{equation}
and consider the composite target
\begin{equation}
T := (S_1, X),
\end{equation}
which Finn and Lizier~\cite{finn2018pointwise} also used in an earlier inconsistency result (see Section~\ref{sec:relation}). Note that the variables $S_1, S_2, X$ are pairwise independent. If $n>2$, the remaining sources $S_3,\ldots,S_n$ can be arbitrary: the argument below involves only $S_1$, $S_2$, and the target, and Lemma~\ref{lem:pw_bounds} holds for every $n$.

\emph{Step 1: Chaining through $S_1$ first yields zero.}
Applying \textbf{TCR} to $T = (S_1, X)$ gives
\begin{equation}
I_\cap(\{1\},\{2\};T) = I_\cap(\{1\},\{2\};S_1) + I_\cap(\{1\},\{2\};X \mid S_1).
\end{equation}
By Lemma~\ref{lem:pw_bounds} and the independence of $S_1$ and $S_2$,
\begin{equation}
I_\cap(\{1\},\{2\};S_1) \le I(S_2;S_1) = 0,
\end{equation}
and by the conditional part of Lemma~\ref{lem:pw_bounds},
\begin{equation}
I_\cap(\{1\},\{2\};X \mid S_1) \le I(S_1;X \mid S_1) = 0.
\end{equation}
Since redundancies are non-negative under \textbf{LP}, both terms vanish, and hence
\begin{equation}\label{eq:order1_zero}
I_\cap(\{1\},\{2\};T) = 0.
\end{equation}

\emph{Step 2: Chaining through $X$ first yields at least one bit.}
Consider the target with permuted components,
\begin{equation}
T' := (X, S_1) = \psi(T), \qquad \psi(a,b) := (b,a).
\end{equation}
Since $\psi$ is a bijection on the support of $T$, \textbf{REI} gives
\begin{equation}\label{eq:rei_swap}
I_\cap(\{1\},\{2\};T') = I_\cap(\{1\},\{2\};T).
\end{equation}
Applying \textbf{TCR} to $T'$ in the order $(X, S_1)$,
\begin{equation}
I_\cap(\{1\},\{2\};T') = I_\cap(\{1\},\{2\};X) + I_\cap(\{1\},\{2\};S_1 \mid X).
\end{equation}
The first term vanishes by Lemma~\ref{lem:pw_bounds}, since $X$ is independent of $S_1$:
\begin{equation}
I_\cap(\{1\},\{2\};X) \le I(S_1;X) = 0.
\end{equation}

For the second term, we use the lower bound from Lemma~\ref{lem:pw_bounds}, applied
with target variable $S_1$ and conditioning variable $X$:
\begin{equation}
I_\cap(\{1\},\{2\};S_1 \mid X)
\ge
I(S_1;S_1\mid X)
+
I(S_2;S_1\mid X)
-
I(S_1,S_2;S_1\mid X).
\end{equation}
We now evaluate the three terms on the right-hand side. Since $X$ is independent of
$S_1$, we have
\begin{equation}
I(S_1;S_1\mid X)=H(S_1\mid X)=1.
\end{equation}
Moreover, conditional on $X=x$, the two sources determine each other via $S_2=S_1\oplus x$ so that
\begin{equation}
I(S_2;S_1\mid X)=1.
\end{equation}
Finally,
\begin{equation}
I(S_1,S_2;S_1\mid X)=H(S_1\mid X)=1.
\end{equation}
Therefore
\begin{equation}
I_\cap(\{1\},\{2\};S_1 \mid X)
\ge
1+1-1
=
1.
\end{equation}

Therefore $I_\cap(\{1\},\{2\};T') \ge 1$, contradicting \eqref{eq:order1_zero} together
with \eqref{eq:rei_swap}.
\end{proof}
\begin{prop}\label{prop:perm_inv}
Theorem~\ref{thm:lp_tcr_incompatible} remains true when \textbf{REI} is replaced by invariance under permutations of the target components, since its proof uses \textbf{REI} only for such a permutation in~\eqref{eq:rei_swap}.
\end{prop}

This proposition isolates the minimal symmetry assumption behind the contradiction. The proof does not require the full strength of \textbf{REI}, but only that permuting the components of a target leaves the PID unchanged. Equivalently, one could build this symmetry directly into an order-independent formulation of \textbf{TCR}. We keep it separate here in order to make the underlying invariance principle explicit. The mechanism of the inconsistency is then particularly transparent: Under \textbf{LP}, chaining through $S_1$ first and chaining through $X$ first yield different values, even though the corresponding targets differ only by a permutation of their components. Thus, any PID satisfying \textbf{LP} and \textbf{TCR} must violate target-permutation invariance, and hence \textbf{REI}.

Just as re-encoding invariance can be weakened to invariance under permutations of the target components (Proposition~\ref{prop:perm_inv}), local positivity can be weakened to two elementary requirements on pairwise redundancies.

\begin{property}[No Information, No Redundancy (\textbf{NR})]\label{property:NR}
If $S_i$ is independent of $T$, then $I_\cap(\{i\},\{j\};T)=0$ for all $j$.
\end{property}

\begin{property}[Full Information, Full Redundancy (\textbf{FR})]\label{property:FR}
If $T$ is a function of $S_i$ and also a function of $S_j$, then $I_\cap(\{i\},\{j\};T)=H(T)$.
\end{property}

A source that has no information about the target cannot share any of it, and two sources that each have complete information about the target share all of it, namely $H(T)$. Both properties fix redundancy only at these two extremes and say nothing about the sign of the remaining atoms.

\begin{prop}\label{prop:nr_fr}
Theorem~\ref{thm:lp_tcr_incompatible} remains true when \textbf{LP} is replaced by \textbf{NR} and \textbf{FR}.
\end{prop}
\begin{proof}
The proof of Theorem~\ref{thm:lp_tcr_incompatible} uses \textbf{LP} only to evaluate four pairwise redundancies. In Step~1, $I_\cap(\{1\},\{2\};S_1)=0$ by \textbf{NR}, since $S_2$ is independent of $S_1$, and $I_\cap(\{1\},\{2\};X\mid S_1)=0$ by \textbf{NR}, since given $S_1=s_1$ the source $S_1$ is constant and hence independent of $X$. In Step~2, $I_\cap(\{1\},\{2\};X)=0$ by \textbf{NR}, since $S_1$ is independent of $X$, and $I_\cap(\{1\},\{2\};S_1\mid X)=1$ by \textbf{FR}, since given $X=x$ both sources determine $S_1$, which is a uniform bit. The conditional redundancies are averages of the same functional under the conditional distributions (Definition~\ref{def:cond_red}), so \textbf{NR} and \textbf{FR} apply to them. Hence $I_\cap(\{1\},\{2\};T)=0$ and $I_\cap(\{1\},\{2\};T')=1$, contradicting~\eqref{eq:rei_swap}.
\end{proof}

\textbf{NR} and \textbf{FR} are indeed weaker than \textbf{LP}: Under \textbf{LP}, redundancies are non-negative and the upper bound of Lemma~\ref{lem:pw_bounds} vanishes if $S_i$ is independent of $T$; if $T$ is a function of both $S_i$ and $S_j$, the lower and upper bounds of Lemma~\ref{lem:pw_bounds} both equal $H(T)$. The converse fails: Ince's $I_{ccs}$~\cite{ince2017measuring} violates \textbf{LP} (Table~\ref{tab:measures_comparison}) but satisfies \textbf{NR} and \textbf{FR}.\footnote{$I_{ccs}$ averages local co-information terms under a maximum-entropy distribution that preserves the joint distribution of each source with the target, keeping only terms whose signs agree with those of the local mutual informations involved~\cite{ince2017measuring}. If $S_i$ is independent of $T$, its local mutual information with $T$ vanishes, so no term is kept. If $T$ is a function of both sources, every term equals $-\log p(t)$, so $I_{ccs}=H(T)$.} Since $I_{ccs}$ also satisfies \textbf{REI}, Proposition~\ref{prop:nr_fr} implies that it violates \textbf{TCR}.

We now return to Theorem~\ref{thm:lp_tcr_incompatible} itself. None of the three properties can be dropped from Theorem~\ref{thm:lp_tcr_incompatible}: any two of them can be satisfied simultaneously.

\begin{prop}[Minimality]\label{prop:minimality}
For every $n \geq 2$ and each pair of the properties \textbf{LP}, \textbf{TCR}, and \textbf{REI}, there exists a PID $\{\Pi(f;T)\}_{f \in \mathcal{B}_n}$ and associated measure of redundant information that satisfy both properties of the pair.
\end{prop}
\begin{proof}
\textbf{LP} and \textbf{REI} are satisfied by the PID induced by $I_{\min}$~\cite{williams2010nonnegative}, and \textbf{TCR} and \textbf{REI} by the PID induced by $I_\cap^{sx}$~\cite{makkeh2021introducing} (see Table~\ref{tab:measures_comparison}). A PID satisfying \textbf{LP} and \textbf{TCR} is constructed in Appendix~\ref{app:chained_pid}.
\end{proof}

Hence \textbf{LP}, \textbf{TCR}, and \textbf{REI} form a minimal inconsistent set of properties. By Theorem~\ref{thm:lp_tcr_incompatible} and Proposition~\ref{prop:perm_inv}, the PID constructed in Appendix~\ref{app:chained_pid} necessarily violates \textbf{REI}: its values depend on the order in which the components of the target are listed.

\subsection{Incompatibility of Identity Property, Local Positivity, and Re-encoding Invariance}\label{sec:id_result}
We here strengthen the incompatibility result of Rauh et al. \cite{rauh2014reconsidering} and present it in our terminology and notation.

The proof analyzes the XOR-Source-Copy Gate \cite{Bertschinger2013_shared_info}.
\begin{definition}[XOR-Source-Copy Gate]\label{def:xor_gate}
The XOR-Source-Copy Gate is defined by two independent binary sources $S_1,S_2 \sim \mathrm{Bern}(1/2)$ and a third source
\begin{equation}
S_3 = S_1 \oplus S_2,
\end{equation}
where $\oplus$ denotes addition modulo $2$. The target variable is the triple
\begin{equation}
T = (S_1,S_2,S_3).
\end{equation}
\end{definition}
The inconsistency proof focuses on the pairwise redundancies in this context
\begin{equation}
I_\cap(\{1\},\{2\};T), \quad I_\cap(\{1\},\{3\};T), \quad I_\cap(\{2\},\{3\};T).
\end{equation}
The key idea is that under \textbf{LP} at least one of these quantities must be strictly positive, while \textbf{IID} and \textbf{REI} force all of them to vanish.

To cover every fixed $n \geq 3$, we allow arbitrary additional sources $S_4,\ldots,S_n$ alongside the XOR-Source-Copy Gate, leaving the target $T=(S_1,S_2,S_3)$ unchanged. We first show that \textbf{LP} forces at least one pairwise redundancy among the original three sources to be strictly positive, irrespective of the additional sources. Figure~\ref{fig:overlap_lemma} illustrates the intuition behind this.

\begin{figure}
    \centering
    \includegraphics{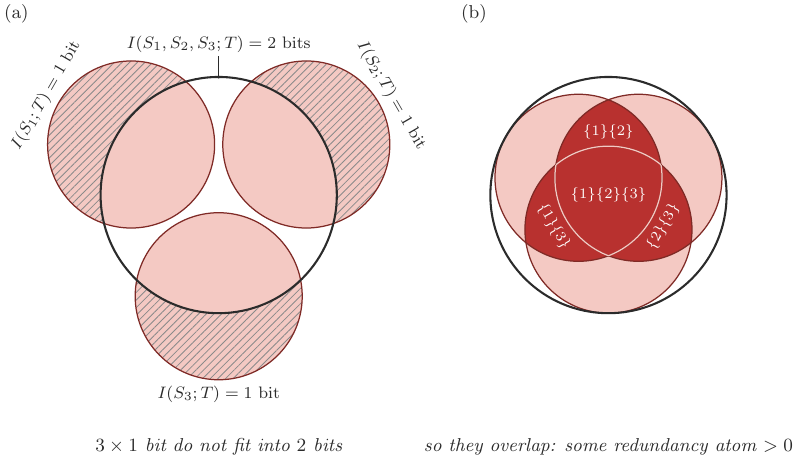}
    \caption{Intuition behind Lemma~\ref{lem:LP_pw_red}. Areas represent amounts of information, and each circle $I(S_i;T)$ has half the area of the joint circle $I(S_1,S_2,S_3;T)$. (a) Placed without overlapping, the three circles do not fit into the joint circle; the hatched parts lie outside it. (b) Inside the joint circle the circles must therefore overlap. The overlap consists of the redundancy atoms $\Pi(\{1\}\{2\})$, $\Pi(\{1\}\{3\})$, $\Pi(\{2\}\{3\})$ and $\Pi(\{1\}\{2\}\{3\})$, so at least one of them is positive. Each of these atoms is part of a pairwise redundancy $I_\cap(\{i\},\{j\};T)$, so at least one pairwise redundancy is positive.}
    \label{fig:overlap_lemma}
\end{figure}

\begin{Lemma}\label{lem:LP_pw_red}
Let $n \geq 3$, let $S_1,S_2,S_3,T$ be distributed according to the XOR-Source-Copy Gate, and let $S_4,\ldots,S_n$ be arbitrary additional sources.
Then \textbf{LP} implies
\begin{equation}
I_\cap(\{i\},\{j\};T) > 0 \quad \text{for some distinct } i,j \in \{1,2,3\}.
\end{equation}
\end{Lemma}

\begin{proof}
We give the proof for $n=3$, i.e.\ without additional sources; the general case is proved in Appendix~\ref{app:ext_xorsc_lemma}.
Consider the quantity
\begin{equation}
\mathrm{RSI} := \sum_{i=1}^3 I(S_i;T)- I(S_1,S_2,S_3;T) = 3 - 2 = 1,
\end{equation}
known from the literature as the \emph{redundancy–synergy index}. Expanding each mutual information term via the consistency equations \eqref{eq:consistency} gives (written in terms of antichain notation explained in Section \ref{sec:background:pid})
\begin{align} \label{eq:rsi_degr_red}
\mathrm{RSI} = \sum_{i=1}^3 \ \sum_{\alpha:\,\{i\}\in \alpha} \Pi(\alpha) - \sum_{\alpha \in \mathcal{A}_3} \Pi(\alpha).
\end{align}
The multiplicity with which an atom $\Pi(\alpha)$ appears in the first sum is determined by its \emph{degree of redundancy} \cite{gutknecht2025shannon},
\begin{equation}
r(\alpha) := \bigl|\{ i \in [n] \mid \{i\} \in \alpha \}\bigr|.
\end{equation}
That is, $r(\alpha)$ counts how many singletons are contained in $\alpha$ and hence how many mutual information terms $I(S_i;T)$ the atom $\Pi(\alpha)$ contributes to. For illustration see Figure \ref{fig:degr_red}.  With this notation the expression simplifies to
\begin{equation}\label{eq:rsi_degr_red_2}
\mathrm{RSI} = \sum_{\alpha:\,r(\alpha)\geq 2} (r(\alpha)-1)\,\Pi(\alpha) - \sum_{\alpha:\,r(\alpha)=0} \Pi(\alpha),
\end{equation}
since all atoms with $r(\alpha)=1$ cancel. For the XOR-Source-Copy Gate this yields
\begin{equation}\label{eq:rsi_xor}
\mathrm{RSI} = \Pi(\{1\}\{2\}) + \Pi(\{1\}\{3\}) + \Pi(\{2\}\{3\}) + 2\Pi(\{1\}\{2\}\{3\}) - \sum_{\alpha:\,r(\alpha)=0} \Pi(\alpha) = 1.
\end{equation}
For this to be true under \textbf{LP}, at least one of the atoms
\[
\Pi(\{1\}\{2\}), \quad \Pi(\{1\}\{3\}), \quad \Pi(\{2\}\{3\}), \quad \Pi(\{1\}\{2\}\{3\})
\]
must be positive. Each of these atoms contributes to at least one pairwise redundancy term $I_\cap(\{i\},\{j\};T)$, so at least one such redundancy must be positive as well under \textbf{LP}.
\end{proof}

\begin{figure}
    \centering
    \includegraphics[width=0.5\linewidth]{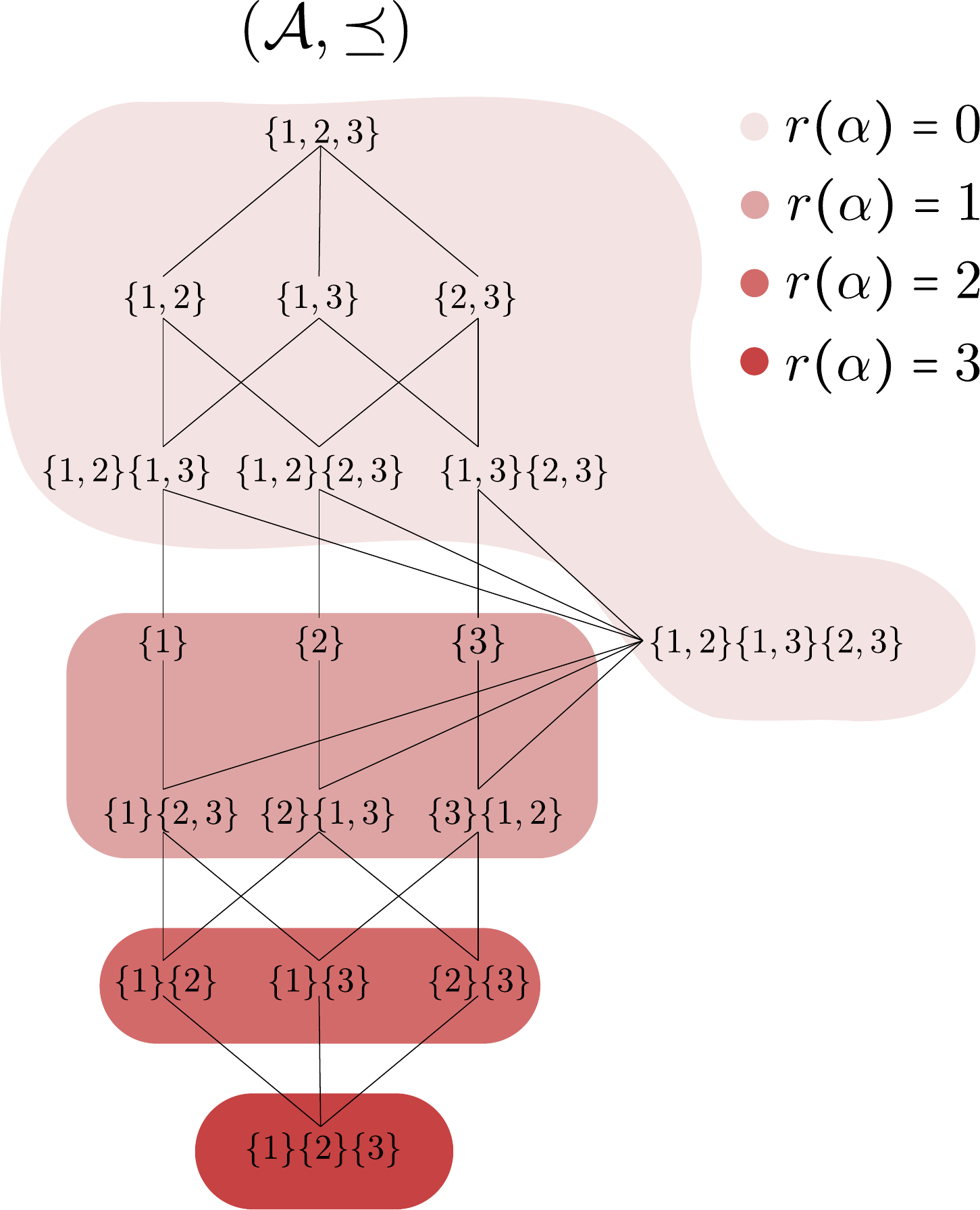}
    \caption{Illustration of the degree of redundancy $r(\alpha)$ of information atoms for $n=3$. The degree of redundancy measures how many individual mutual information terms $I(S_i;T)$ an atom contributes to, and hence, how often it appears in the first term in Equation \eqref{eq:rsi_degr_red}. The atoms with $r(\alpha)=0$ are synergistic in the sense that they cannot be obtained from any individual information source. These only appear in the second term of Equation \eqref{eq:rsi_degr_red} and hence end up with a minus sign in Equation \eqref{eq:rsi_degr_red_2}. The atoms with $r(\alpha)=1$ appear exactly once in both terms and hence cancel. The atoms with $r(\alpha)\geq 2$, one might say the genuinely redundant ones, end up with a plus sign and multiplicity $r(\alpha)-1$ in Equation \eqref{eq:rsi_degr_red_2}.}
    \label{fig:degr_red}
\end{figure}

\begin{thm}[Incompatibility of \textbf{LP}, \textbf{IID}, and \textbf{REI}]\label{thm:id_lp_rei_incomp}
For every $n\geq 3$, there exists no PID $\{\Pi(f;T)\}_{f \in \mathcal{B}_n}$ and associated measure of redundant information $\{I_\cap(\alpha;T)\}_{\alpha \in \mathcal{A}_n}$ that simultaneously satisfy
local positivity, the independent identity property, and re-encoding invariance.
\end{thm}
\begin{proof} Assume for contradiction that such a PID exists, let $S_1,S_2,S_3,T$ be distributed according to the XOR-Source-Copy Gate, and let $S_4,\ldots,S_n$ be arbitrary additional sources; throughout the proof, the indices $i\neq j$ range over $\{1,2,3\}$. By Lemma \ref{lem:LP_pw_red}, at least one pairwise redundancy $I_\cap(\{i\},\{j\};T)$ must be strictly positive:
\begin{equation}\label{eq:LP_pwr}
I_\cap(\{i\},\{j\};T) > 0    \quad \text{for some } i\neq j.
\end{equation}
Now, we show that under \textbf{IID} and \textbf{REI} they must all vanish. This part of the argument involves only the sources $S_1,S_2,S_3$ and the target $T$, whose joint distribution is unaffected by the presence of the additional sources.

The sources $S_1,S_2,S_3$ are pairwise independent, hence by \textbf{IID} we obtain
\begin{equation} \label{eq:pwr_zero}
I_\cap(\{i\},\{j\};(S_i,S_j)) = 0 \quad \text{for all } i\neq j.
\end{equation}
Moreover, the target $T$ is an invertible function of the sources
\begin{equation}
T = (S_1,S_2,S_3) = F_{ij}(S_i,S_j),
\end{equation}
for $i\neq j$ and where $F_{ij}$ is a bijection from $\{0,1\}^2$ onto $\operatorname{supp}(T) = \{(0,0,0),(0,1,1),(1,0,1),(1,1,0)\}$.
For example, with $(i,j)=(1,2)$,
\begin{align}
(0,0) &\mapsto (0,0,0), \\
(0,1) &\mapsto (0,1,1), \\
(1,0) &\mapsto (1,0,1), \\
(1,1) &\mapsto (1,1,0).
\end{align}
By \textbf{REI}, redundant information is invariant under such invertible re-encodings of the target. Hence
\begin{equation}
I_\cap(\{i\},\{j\};T) = I_\cap(\{i\},\{j\};(S_i,S_j)) \overset{\eqref{eq:pwr_zero}}{=} 0 \quad \text{for all } i\neq j,
\end{equation}
which establishes the claim and yields a contradiction with \eqref{eq:LP_pwr}, thus establishing the inconsistency.
\end{proof}

\begin{cor}\label{cor:id_lp_rei}
In particular, \textbf{LP}, \textbf{ID}, and \textbf{REI} are inconsistent, since \textbf{ID} implies \textbf{IID}.
\end{cor}

Compared with the original result of Rauh et al.~\cite{rauh2014reconsidering}, Theorem~\ref{thm:id_lp_rei_incomp} differs in the following respects. First, the identity property is only needed for pairs of independent single sources, i.e.\ in the form \textbf{IID}. The original proof relies on a more general version of the identity axiom for collections of sources~\cite{Bertschinger2013_shared_info}. Second, our proof does not use source monotonicity (\textbf{SM}), which enters the original argument twice but already follows from \textbf{LP}. Third, \textbf{REI} is stated as an explicit assumption whereas the original proof just assumes it implicitly. Finally, Rauh et al.\ state in their conclusions that the incompatibility holds for all $n\geq 3$, while their proof treats only $n=3$. Lemma~\ref{lem:LP_pw_red} supplies the formal argument for general $n$.

As for Theorem~\ref{thm:lp_tcr_incompatible}, none of the three properties can be dropped from Theorem~\ref{thm:id_lp_rei_incomp}: any two of them can be satisfied simultaneously.

\begin{prop}[Minimality]\label{prop:minimality_id}
For every $n \geq 3$ and each pair of the properties \textbf{LP}, \textbf{ID}, and \textbf{REI}, there exists a PID $\{\Pi(f;T)\}_{f \in \mathcal{B}_n}$ and associated measure of redundant information that satisfy both properties of the pair.
\end{prop}
\begin{proof}
\textbf{LP} and \textbf{REI} are satisfied by the PID induced by $I_{\min}$~\cite{williams2010nonnegative}, and \textbf{ID} and \textbf{REI} by the PID induced by $I^{\prec}_\cup$~\cite{Kolchinsky2022} (see Table~\ref{tab:measures_comparison}). The chained PID of Appendix~\ref{app:chained_pid} satisfies \textbf{LP} and \textbf{ID} (Proposition~\ref{prop:chained_pid}).
\end{proof}

Hence \textbf{LP}, \textbf{ID}, and \textbf{REI} also form a minimal inconsistent set of properties for every $n\geq 3$. Since \textbf{ID} implies \textbf{IID}, the same examples show that \textbf{LP}, \textbf{IID}, and \textbf{REI} form a minimal inconsistent set as well.

The
incompatibility of \textbf{LP}, \textbf{ID}, and \textbf{REI} does not arise for $n=2$:
the BROJA measure \cite{bertschinger2014quantifying} satisfies all three properties
simultaneously in this case, so Theorem~\ref{thm:id_lp_rei_incomp} cannot be extended
to two sources.

\subsection{\texorpdfstring{Relation between the two inconsistency results}{Relation between the two inconsistency results}}\label{sec:relation}

The two inconsistency results are closely connected: under \textbf{LP}, the target chain rule implies the identity property. This implication has been observed by Bertschinger et al.~\cite{Bertschinger2013_shared_info} (Section~4 there), but without proof. We provide a proof here, since the implication clarifies how Theorems~\ref{thm:lp_tcr_incompatible} and~\ref{thm:id_lp_rei_incomp} are related.

\begin{Lemma}[\textbf{LP} and \textbf{TCR} imply \textbf{ID}]\label{lem:lp_tcr_id}
Assume \textbf{LP} and \textbf{TCR}. Then for every $n \geq 2$ and all distinct $i \neq j$,
\begin{equation}
I_\cap(\{i\},\{j\};(S_i,S_j)) = I(S_i;S_j),
\end{equation}
that is, \textbf{ID} holds for every pair of sources, irrespective of the remaining sources.
\end{Lemma}

\begin{proof}
Applying \textbf{TCR} to the target $(S_i,S_j)$ gives
\begin{equation}
I_\cap(\{i\},\{j\};(S_i,S_j))
= I_\cap(\{i\},\{j\};S_i) + I_\cap(\{i\},\{j\};S_j \mid S_i).
\end{equation}
For both terms, the lower and upper bounds of Lemma~\ref{lem:pw_bounds} coincide.

\emph{First term.} With target $S_i$ we have $I(S_i;S_i) = I(S_i,S_j;S_i) = H(S_i)$. The lower bound therefore equals $H(S_i) + I(S_j;S_i) - H(S_i) = I(S_i;S_j)$, and the upper bound equals $\min\{H(S_i),\, I(S_i;S_j)\} = I(S_i;S_j)$. Hence
\begin{equation}
I_\cap(\{i\},\{j\};S_i) = I(S_i;S_j).
\end{equation}

\emph{Second term.} Conditioned on $S_i$, the source $S_i$ is constant, so $I(S_i;S_j \mid S_i) = 0$, while $I(S_j;S_j \mid S_i) = I(S_i,S_j;S_j \mid S_i) = H(S_j \mid S_i)$. The conditional lower bound therefore equals $0 + H(S_j \mid S_i) - H(S_j \mid S_i) = 0$, and the conditional upper bound equals $\min\{0,\, H(S_j \mid S_i)\} = 0$. Hence
\begin{equation}
I_\cap(\{i\},\{j\};S_j \mid S_i) = 0.
\end{equation}

Adding both terms yields the claim.
\end{proof}

Lemma~\ref{lem:lp_tcr_id} yields a second route to Theorem~\ref{thm:lp_tcr_incompatible} for $n\geq 3$: a PID satisfying \textbf{LP}, \textbf{TCR}, and \textbf{REI} would also satisfy \textbf{LP}, \textbf{IID}, and \textbf{REI}, contradicting Theorem~\ref{thm:id_lp_rei_incomp}; only the case of independent sources in Lemma~\ref{lem:lp_tcr_id} is needed. For $n=2$, this route is not available. The BROJA measure satisfies \textbf{LP}, \textbf{ID}, and \textbf{REI} (Section~\ref{sec:id_result}), so the two-source case of Theorem~\ref{thm:lp_tcr_incompatible} cannot be derived from an inconsistency of \textbf{LP}, \textbf{ID}, and \textbf{REI}.

Lemma~\ref{lem:lp_tcr_id} also clarifies the relation of Theorem~\ref{thm:lp_tcr_incompatible} to the result of Finn and Lizier~\cite{finn2018pointwise}, who showed that \textbf{LP}, \textbf{TCR}, \textbf{ID}, and \textbf{REI} are incompatible, although \textbf{REI} is not stated as an assumption there. Their proof uses the same sources and composite target as the proof of Theorem~\ref{thm:lp_tcr_incompatible}. Finn and Lizier saw their result as highlighting difficulties introduced by the identity property. Together with Lemma~\ref{lem:lp_tcr_id}, however, it also yields the inconsistency of \textbf{LP}, \textbf{TCR}, and \textbf{REI}, so giving up the identity property cannot resolve the conflict. Moreover, their argument applies \textbf{REI} to the re-encoding $(S_1,S_2)\mapsto(S_1,S_1\oplus S_2)$, which replaces a component of the target, whereas Theorem~\ref{thm:lp_tcr_incompatible} applies it only to swapping the two components, which merely changes the order in which they are listed (Proposition~\ref{prop:perm_inv}).

\section{Generality across PID-inducing concepts} \label{sec:generality}

The inconsistency results in the previous section were formulated in terms of redundant information. However, as noted in Section~\ref{sec:background:pid_inducing}, redundancy is just one of several PID-inducing concepts. A natural question is therefore whether the incompatibilities we have identified are artifacts of this particular choice, or whether they persist under alternative formulations. In this section we introduce the general framework of $\mathcal{C}$-information (for full details see \cite{gutknecht2025babel}), define several alternative PID-inducing concepts, and show that the same inconsistencies arise for all of them.

\subsection{\texorpdfstring{$\mathcal{C}$}{C}-information and alternative PID-inducing concepts}\label{sec:cinfo}

The notion of redundant information (Definition~\ref{def:redundancy}) can be generalized by replacing the specific logical condition $\mathcal{C}_\mathrm{red}$ with an arbitrary condition on parthood distributions.

\begin{definition}[$\mathcal{C}$-Information]\label{def:c_info}
Let $\mathcal{C}(\mathbf{a}_1,\ldots,\mathbf{a}_m;f)$ be a logical condition that specifies a relationship between input sets $\mathbf{a}_1,\ldots,\mathbf{a}_m$  and parthood distributions $f \in \mathcal{B}_n$. A real-valued functional $I^\mathcal{C}(\mathbf{a}_1,\ldots,\mathbf{a}_m; T)$ of the joint distribution of $(\mathbf{S}, T)$ is called a measure of $\mathcal{C}$-information if and only if there exists a PID $\{\Pi(f;T)\}_{f \in \mathcal{B}_n}$ such that:
\begin{equation}\label{eq:atoms_c_info}
I^\mathcal{C}(\mathbf{a}_1,\ldots,\mathbf{a}_m; T) = \sum\limits_{\substack{f\in \mathcal{B}_n \\ \mathcal{C}(\mathbf{a}_1,\ldots,\mathbf{a}_m;f)}} \Pi(f;T).
\end{equation}
The condition $\mathcal{C}$ is called \textit{PID-inducing} if the system of linear equations \eqref{eq:atoms_c_info}, with the atoms $\Pi(f;T)$ as unknowns, has exactly one solution, i.e.\ if its coefficient matrix, which depends only on $\mathcal{C}$, has full column rank. Measures of $\mathcal{C}$-information are then called PID-inducing as well.
\end{definition}

Accordingly, we call $\mathcal{C}$-information a PID-inducing concept if $\mathcal{C}$ is PID-inducing; this makes precise the informal notion of Section~\ref{sec:background:pid_inducing}. Redundant information is recovered by setting $\mathcal{C} = \mathcal{C}_\mathrm{red}$. Three further PID-inducing concepts deserve mention because of their particularly simple and interpretable structure.

\begin{definition}[Union Information] \label{def:union}
A real-valued functional $I_{\cup}(\mathbf{a}_1,\ldots,\mathbf{a}_m;T)$ of the joint distribution of $(\mathbf{S},T)$ is called a measure of union information if it is a measure of $\mathcal{C}$-information with respect to the condition
\begin{equation}
\mathcal{C}_{union}(\mathbf{a}_1,\ldots,\mathbf{a}_m;f) \Leftrightarrow \exists i \in [m]: f(\mathbf{a}_i) = 1
\end{equation}
\end{definition}
This measures all information which  \textit{can be obtained from at least one} of the input source collections $\mathbf{a}_1,\ldots, \mathbf{a}_m$.

\begin{definition}[Weak Synergy] \label{def:ws}
A real-valued functional $I_{ws}(\mathbf{a}_1,\ldots,\mathbf{a}_m;T)$ of the joint distribution of $(\mathbf{S},T)$ is called a measure of weak synergy  if it is a measure of $\mathcal{C}$-information with respect to the condition
\begin{equation}
\mathcal{C}_{ws}(\mathbf{a}_1,\ldots,\mathbf{a}_m;f) \Leftrightarrow \forall i \in [m]: f(\mathbf{a}_i) = 0
\end{equation}
\end{definition}
This measures all information which \textit{cannot be obtained from any} individual input source collection among the $\mathbf{a}_1,\ldots, \mathbf{a}_m$ and is hence synergistic with respect to these collections. The meaning of the qualifier ``weak'' comes from the fact that it is the weakest condition out of multiple plausible conditions for describing synergistic information \cite{gutknecht2021bits, gutknecht2025babel}. 

\begin{definition}[Vulnerable Information] \label{def:vul}
A real-valued functional $I_{vul}(\mathbf{a}_1,\ldots,\mathbf{a}_m;T)$ of the joint distribution of $(\mathbf{S},T)$ is called a measure of vulnerable information if it is a measure of $\mathcal{C}$-information with respect to the condition
\begin{equation}
\mathcal{C}_{vul}(\mathbf{a}_1,\ldots,\mathbf{a}_m;f) \Leftrightarrow \exists i \in [m]: f(\mathbf{a}_i)=0
\end{equation}
\end{definition}
This measures all information which \textit{cannot be obtained from at least one} individual input source collection among the $\mathbf{a}_1,\ldots, \mathbf{a}_m$. It is therefore vulnerable to loss: if we lose access to certain sources and are left with only a collection that lacks this information, then the information is lost.

As with redundant information, conditioning is defined by evaluating the same functional under the conditional distribution and averaging:

\begin{definition}[Conditional $\mathcal{C}$-information]
Let $(\mathbf{S},T,Z)$ be jointly distributed random variables, where $Z$ may be an external 
variable or part of $\mathbf{S}$ or $T$. Let $I^\mathcal{C}$ be a measure of $\mathcal{C}$-information. Then we define
\begin{align}
I^\mathcal{C}(\mathbf{a}_1,\ldots,\mathbf{a}_m;T \mid Z) 
&:= \sum_{z} p(z)\, I^{\mathcal{C}, p(\mathbf{S},T \mid Z=z)}(\mathbf{a}_1,\ldots,\mathbf{a}_m;T),
\end{align}
where the superscript $p(\mathbf{S},T \mid Z=z)$ indicates that we are evaluating the same functional $I^\mathcal{C}$ under the conditional distribution of $(\mathbf{S},T)$ given $Z=z$.
\end{definition}

\subsection{Equivalent formulations of inconsistency results}\label{sec:equiv_formulations}

We now show that the properties \textbf{TCR} and \textbf{ID}, which featured in our inconsistency results and were formulated in terms of redundant information, admit equivalent formulations for other PID-inducing concepts. The incompatibilities established in Theorems~\ref{thm:lp_tcr_incompatible} and~\ref{thm:id_lp_rei_incomp} therefore persist under these reformulations.

\begin{prop}[Equivalent versions of \textbf{ID}]\label{prop:equiv_id}
Let $T= (S_1,S_2)$, where further sources $S_3,\ldots,S_n$ may be present, and let $I_\cap$, $I_\cup$, $I_{ws}$, and $I_{vul}$ be the measures of redundant, union, weakly synergistic, and vulnerable information associated with the same PID. Then the following are equivalent:
\begin{enumerate}
    \item[(i)] \textbf{ID} as defined in Property \ref{property:ID}.
    \item[(ii)] $I_\cup(\{1\},\{2\};T) = H(S_1,S_2)$
    \item[(iii)] $I_{ws}(\{1\},\{2\};T) = 0$
    \item[(iv)] $I_{vul}(\{1\},\{2\};T) = H(S_1|S_2) + H(S_2|S_1)$
\end{enumerate}
In particular, if $S_1$ and $S_2$ are independent, \textbf{IID} is equivalent to each of $I_\cup(\{1\},\{2\};T) = H(S_1)+H(S_2)$, $I_{ws}(\{1\},\{2\};T) = 0$, and $I_{vul}(\{1\},\{2\};T) = H(S_1)+H(S_2)$.
\end{prop}
\begin{proof}
See Appendix \ref{app:equiv_id}.
\end{proof}

The idea behind the following proposition is simple. Suppose that a particular PID-inducing measure, for instance a particular measure of redundant information, satisfies a target chain rule. Since the atoms it induces are fixed linear combinations of its values, these atoms satisfy a target chain rule as well. From the atoms, a measure of $\mathcal{C}$-information can be assembled for every condition $\mathcal{C}$ by summing the atoms that satisfy $\mathcal{C}$ (Definition~\ref{def:c_info}), and all of these associated measures inherit the chain rule. And to complete the circle: if all measures of $\mathcal{C}$-information associated with a PID satisfy a target chain rule, then in particular each single one of them does.
\begin{prop}[Equivalent versions of \textbf{TCR}]\label{prop:equiv_tcr}
Let $\{\Pi(f;T)\}_{f \in \mathcal{B}_n}$ be a PID and let $T=(T_1,T_2)$. The following are equivalent:
\begin{enumerate}
\item[(i)] For some PID-inducing condition $\mathcal{C}$, the measure of $\mathcal{C}$-information inducing the PID $\{\Pi(f;T)\}_{f \in \mathcal{B}_n}$ satisfies a target chain rule (\textbf{C-TCR}):
\begin{equation}
I^{\mathcal{C}}(\mathbf{a}_1,\ldots,\mathbf{a}_m;T) \;=\; I^{\mathcal{C}}(\mathbf{a}_1,\ldots,\mathbf{a}_m;T_1) \;+\; I^{\mathcal{C}}(\mathbf{a}_1,\ldots,\mathbf{a}_m;T_2 \mid T_1).
\end{equation}
\item[(ii)] The atoms satisfy a target chain rule (\textbf{A-TCR}):
\begin{equation}
\Pi(f;T) \;=\; \Pi(f;T_1) \;+\; \Pi(f;T_2 \mid T_1).
\end{equation}
\item[(iii)] For every condition $\mathcal{C}$, the measure of $\mathcal{C}$-information associated with these atoms satisfies \textbf{C-TCR}.
\end{enumerate}
In particular, since redundant information is PID-inducing, \textbf{TCR} for the measure of redundant information inducing a given PID is equivalent to the target chain rule for the associated measure of any other PID-inducing concept, such as union information (\textbf{U-TCR}), weak synergy (\textbf{S-TCR}), or vulnerable information (\textbf{V-TCR}).
\end{prop}

\begin{proof}
See Appendix \ref{app:equiv_tcr}.
\end{proof}

These equivalences immediately yield the following corollaries of our main theorems.

\begin{cor}\label{cor:thm1_equiv}
The incompatibility of \textbf{LP}, \textbf{IID}, and \textbf{REI} (Theorem~\ref{thm:id_lp_rei_incomp}) still holds when \textbf{IID} is replaced by any of its equivalent formulations in Proposition~\ref{prop:equiv_id}.
\end{cor}

\begin{cor}\label{cor:thm2_equiv}
The incompatibility of \textbf{LP}, \textbf{TCR}, and \textbf{REI} (Theorem~\ref{thm:lp_tcr_incompatible}) still holds when \textbf{TCR} is replaced by any of its equivalent formulations in Proposition~\ref{prop:equiv_tcr}. 
\end{cor}

In particular, replacing \textbf{TCR} with the atom-level target chain rule \textbf{A-TCR} (Proposition~\ref{prop:equiv_tcr}(ii)) yields a formulation that is entirely a statement about 
the information atoms: there is no collection of atoms satisfying the consistency 
equations~\eqref{eq:consistency} that simultaneously satisfies \textbf{LP}, \textbf{A-TCR}, 
and \textbf{REI}. In this formulation, no redundancy function or other PID-inducing concept 
is involved.

\section{Discussion} \label{sec:discussion}
Our main result identifies a fundamental obstruction to extending the familiar behaviour of mutual information to PID. For every fixed number of sources $n\geq 2$, no PID can simultaneously satisfy local positivity, the target chain rule, and re-encoding invariance. In fact, our main result requires only a particularly weak special case of re-encoding invariance, namely invariance under permutations of the components of a multivariate target. Importantly, this inconsistency does not depend on any specifically PID-motivated axiom such as the identity property. Instead, it shows that PID cannot simply inherit the basic structural properties of mutual information: its non-negativity, its chain rule, and its invariance under invertible transformations. At least one of these classical properties must be abandoned in the PID setting.

Our second incompatibility result concerns a different, more PID-specific tension. Building on the theorem of Rauh et al.~\cite{rauh2014reconsidering}, we show that local positivity, the independent identity property, and re-encoding invariance cannot be jointly satisfied for any fixed $n\geq 3$, even without assuming source monotonicity. Unlike our main inconsistency result, however, this obstruction does not extend to the special case $n=2$, since the BROJA measure satisfies all three properties in that setting.
 Table~\ref{tab:measures_comparison} shows how prominent PID proposals in the literature navigate the trade-offs between the properties \textbf{LP}, \textbf{TCR}, \textbf{REI}, \textbf{ID}, and \textbf{IID}.

 \begin{table}[!ht]
     \centering
     \begin{tabular}{c|c|c|c|c|c|c}
       Measure   & LP & TCR & REI & ID & IID & $n>2$\\\hline
        $I_{min}$~\cite{williams2010nonnegative} 
         & \checkmark & X$^{a}$ & \checkmark & X & X &\checkmark \\
        $I_\cap^{sx}$~\cite{makkeh2021introducing}& 
         X & \checkmark & \checkmark &X$^{d}$ & X$^{d}$ & \checkmark \\
         BROJA~\cite{bertschinger2014quantifying}& \checkmark & X$^{a}$ & \checkmark &\checkmark & \checkmark & X\\
         $R_{min}$~\cite{finn2018pointwise}& X &\checkmark & \checkmark& X & X & \checkmark \\
         $I_{ccs}$~\cite{ince2017measuring}& X$^{b}$ & X$^{c}$ & \checkmark & X & \checkmark & \checkmark \\
         $I_{red}$~\cite{Harder2013}& \checkmark & X$^{a}$ & \checkmark& \checkmark& \checkmark &X \\
         $I^{\prec}_\cup$~\cite{Kolchinsky2022}& \checkmark\ ($n=2$), X$^{b}$ ($n\geq 3$) & X$^{a}$ & \checkmark & \checkmark & \checkmark & \checkmark\\
     \end{tabular}
    
     \caption{Properties satisfied (\checkmark) or violated (X) by prominent PID measures; the column ``$n>2$'' indicates whether a measure is defined for more than two sources. Unmarked entries follow the cited literature; marked entries follow from: $^{a}$~Theorem~\ref{thm:lp_tcr_incompatible}, as the measure satisfies \textbf{LP} and \textbf{REI} ($I^{\prec}_\cup$ for $n=2$, where it coincides with BROJA); $^{b}$~Theorem~\ref{thm:id_lp_rei_incomp} for $n\geq 3$, as the measure satisfies \textbf{IID} and \textbf{REI}; $^{c}$~Proposition~\ref{prop:nr_fr}, as $I_{ccs}$ satisfies \textbf{NR}, \textbf{FR}, and \textbf{REI} (Section~\ref{sec:main_result}); $^{d}$~the two-bit copy of independent fair bits, to which $I_\cap^{sx}$ assigns $\log_2(4/3)\approx 0.415$ bits of redundancy, violating \textbf{IID} and hence \textbf{ID}. Measures satisfying \textbf{ID} also satisfy \textbf{IID}. All measures satisfy \textbf{REI}, since they depend only on the probabilities of outcomes, not on their labels. For $I^{\prec}_\cup$, the entries refer to the PID it induces as a measure of union information (Definition~\ref{def:union}), not to Kolchinsky's own decomposition, which does not satisfy the consistency equations (see ``Rejecting the consistency equations?'' below).}
     \label{tab:measures_comparison}
 \end{table}

Before turning to the broader implications, we note three limitations concerning the scope of these results. First, as mentioned above, the inconsistency involving the identity property does not hold for $n=2$. Second, both impossibility results do not exclude the possibility that a measure satisfies the relevant properties on certain restricted classes of joint distributions.  Third, our analysis is confined to discrete random variables, as our counterexamples are discrete. While this suffices to rule out measures satisfying the incompatible properties for all discrete distributions, and hence also in complete generality (covering both discrete and continuous cases), it leaves open whether measures restricted solely to continuous distributions might avoid these inconsistencies. These limitations notwithstanding, our results establish fundamental barriers to any PID framework that aspires to full generality.

Since the properties in each incompatible set cannot be preserved at once in the general case, we are forced to make a choice: at least one must be abandoned or relaxed. In what follows, we briefly consider each property in turn, examining the case for retaining or rejecting it.

\paragraph{Re-encoding Invariance.}
Among the properties considered in our inconsistency results, re-encoding invariance (REI) appears to be the least negotiable. In information theory, the very notion of “information” is grounded in the idea of \textit{uncertainty reduction} \cite{shannon1948mathematical}.  This reduction should depend only on the abstract probabilistic structure of the situation—on how many possibilities there are and how likely each is—not on how those possibilities are represented. Whether we label four equiprobable outcomes as $\{0,1,2,3\}$, as binary strings $\{00,01,10,11\}$, or as one-hot vectors makes no difference to the uncertainty we face before observing the outcome or to the certainty we gain afterwards. Likewise, our uncertainty about tomorrow’s temperature should not depend on whether it is measured in Celsius or Fahrenheit: it would be incoherent to be quite certain that it will be a warm, sunny day when expressed in one unit, yet highly uncertain when expressed in the other, unsure whether to bring a sunhat or a winter coat. In this sense, re-encoding invariance follows directly from the fundamental interpretation of information as an observer’s uncertainty reduction.

 In the discrete setting, this invariance is guaranteed for any quantity that treats the variables' outcomes only as abstract symbols, using nothing but their probabilities—neither their numerical values nor any internal structure. All standard quantities of information theory are of this kind. A bijective relabeling merely changes the symbols to which the probability values are attached, leaving their probabilistic relationships unchanged. Hence, any measure that depends only on the distribution of probabilities---and not on additional algebraic or geometric structure of the outcome space—will yield identical results across re-encodings. Non-invariance arises only when such additional structure is invoked, for example, when outcomes are treated as numbers that can be added, ordered, or compared by distance. Quantities like the mean, median, or variance exemplify this: rescaling variables alters their values, because they depend on the particular numerical encoding, not just on the underlying probabilities. Importantly, and perhaps less obviously, a composite target $T=(T_1,T_2)$ carries additional structure of exactly this kind: its outcomes are pairs, which can be split into components listed in a particular order. A chain rule relies on this structure explicitly, and a PID can exploit it as well. In fact, by Proposition~\ref{prop:perm_inv}, any PID satisfying \textbf{LP} and \textbf{TCR} must depend on the order in which the components of the target are listed; the chained PID of Appendix~\ref{app:chained_pid} is an example. Even in the discrete case, \textbf{REI} is therefore not automatically satisfied by a PID, but a substantive requirement.

A useful analogy can be drawn with principal component analysis (PCA). PCA is often interpreted as identifying the directions of greatest “uncertainty” or “information” in a dataset: those along which the data vary most. However, PCA is based on variance, which is not re-encoding invariant. As a result, the apparent “uncertainty” revealed by PCA depends strongly on the units and characteristic scale of the variables. A variable measured in large numerical units, or one that naturally exhibits higher variance, will dominate the principal components even if it is not intrinsically more informative. To obtain meaningful results, one must employ workarounds that seek to make variables more comparable such as measuring them in the same units and standardizing them. This example shows how failure of transformation invariance can distort interpretations couched in terms of uncertainty.

It is worth adding a few brief remarks on how re-encoding invariance manifests in the continuous setting. When variables are continuous and related by smooth invertible re-encodings, the corresponding probability densities transform by the determinant of the Jacobian. To remain invariant, an information measure must cancel out this induced change in the density. Differential entropy is the prime counterexample: it fails to do so and is therefore not re-encoding invariant. Jaynes \cite{jaynes2003probability} already emphasized that this non-invariance undermines the use of the standard differential entropy expression as a straightforward continuous analogue of Shannon’s discrete entropy, and argued that an invariant formulation requires introducing an appropriate reference measure $m(x)$ that transforms along with the density. Mutual information, by contrast, is known to be invariant under all smooth invertible re-encodings \cite{kraskov2004estimating}. The requirement of invariance has also been proposed for measures of \textit{specific information}—the information conveyed by a particular realization of one variable about another—and it has been shown that this condition uniquely determines the appropriate functional form \cite{kostal2018coordinate}.

\paragraph{Target Chain Rule.}
The chain rule of mutual information is fundamental to classical information theory. In Fano's influential exposition of information theory~\cite{Fano1961}, it is even taken as an axiom. Carrying it over to the PID realm therefore seems desirable, as it allows us to break reasoning about composite targets into individual steps. A similar idea is already contained in Shannon's grouping axiom for entropy which requires that the uncertainty about a set of outcomes can be written as the uncertainty about which group an outcome belongs to, plus the remaining uncertainty about the outcome once the group is known \cite{shannon1948mathematical}. The chain rule is also a powerful technical tool: most theoretical derivations involving mutual information rely on the chain rule at some point.

On the other hand, several widely used and well-motivated PID proposals violate \textbf{TCR}. For instance, the BROJA measure \cite{bertschinger2014quantifying} has a plausible decision-theoretic interpretation yet does not satisfy \textbf{TCR}. This suggests a basis for rejecting \textbf{TCR}: we may prioritize a measure with a particular operational interpretation over satisfaction of structural properties. When the operational interpretation is sufficiently compelling, violation of \textbf{TCR} may be an acceptable price to pay.

\paragraph{Identity Property.} The identity property stands on somewhat different footing than the other properties we have considered. Unlike \textbf{TCR}, \textbf{LP}, or \textbf{REI}, it does not generalize a fundamental property of mutual information to the PID setting. Rather, it was introduced by Harder et al.~\cite{Harder2013} to address what they saw as a pathological behavior of the $I_{\min}$ measure proposed by Williams and Beer~\cite{williams2010nonnegative}:
\begin{equation}
I_{\min}(\mathbf{a}_1,\ldots,\mathbf{a}_m;T) := \sum_{t} p(t)\, \min_{i} I(\mathbf{a}_i;t),
\end{equation}
where $I(\mathbf{a}_i;t)$ is the \textit{specific information} that source collection $\mathbf{a}_i$ provides about the realization $T=t$:
\begin{equation}
I(\mathbf{a}_i;t) := \sum_{\substack{(s_j)_{j\in \mathbf{a}_i}}}
p\bigl((s_j)_{j\in \mathbf{a}_i} \mid t\bigr)\,
\log \frac{p\bigl(t \mid (s_j)_{j\in \mathbf{a}_i}\bigr)}{p(t)}.
\end{equation}

The motivating example is instructive: let $S_1$ and $S_2$ be independent fair coin flips, and $T=(S_1,S_2)$. Harder et al. argued that the two sources provide entirely distinct information: one reveals only the first component of the target, the other only the second component. Yet $I_{\min}$ assigns 1 bit of redundancy because each source provides the same \emph{amount} of specific information about each target realization (namely, 1 bit). According to Harder et al., this shows that $I_{\min}$ conflates "sameness of amount" with "sameness of information." A measure that adequately accounts for information \emph{content} rather than merely its amount should therefore satisfy the identity property, assigning zero redundancy when sources are independent.

However, this argument is problematic. Whether there is truly zero redundancy in this scenario depends on how one conceptualizes redundant information. Bertschinger et al. \cite{Bertschinger2013_shared_info} argued that if redundancy is understood in terms of \textit{shared knowledge}  in the sense of multi-agent epistemic game theory \cite{fagin2004reasoning}, then the identity property need not hold. Indeed, the shared exclusions measure of redundancy \cite{makkeh2021introducing}, $I_\cap^{sx}$, quantifies shared knowledge in this sense, which explains why it violates ID while maintaining an operational interpretation as "the information about the target provided by the shared knowledge of multiple agents" (this is shown formally in \cite{gutknecht2023information}).

This suggests that while $I_{\min}$ does indeed conflate sameness of amount with sameness of information (this is literally what its definition entails, assigning redundancy whenever specific information values coincide) one can nevertheless violate the identity property without making this conflation. In that sense, the term "identity property" is somewhat of a misnomer: it is arguably neither necessary nor sufficient for doing justice to the distinction between informational amount and content (i.e. identity). 

Overall, since \textbf{ID} addresses behavior in a special case rather than generalizing any core property of information theory, and since there are reasonable arguments that it need not hold for certain plausible conceptions of redundancy, it should not be demanded as a \textit{universal} requirement for acceptable PID measures. Rather, there may exist legitimate conceptualizations of redundancy where it holds (e.g.,~\cite{bertschinger2014quantifying, Harder2013}) and others where it does not (e.g.,~\cite{makkeh2021introducing}). Those that do satisfy \textbf{ID} will of course inevitably have to pay the price of not satisfying \textbf{LP} for $n\geq 3$ (given that we accept \textbf{REI}; see the discussion above).

\paragraph{Local Positivity.}
Non-negativity of information is widely considered desirable, if not essential, in information theory. Certain information measures have been criticized precisely because they can take negative values, interaction information being a prominent example. The requirement also makes intuitive sense if we think of information as quantifying the \textit{strength of statistical dependence}, independent of its nature or direction. In the extreme case where variables are independent, information vanishes; it does not seem meaningful to speak of negative dependence strength.

However, negative values might be acceptable if they can be given a reasonable interpretation, most naturally, as \textit{misinformation}. At the \textit{pointwise} level of individual realizations \cite{Fano1961}, there is indeed a clear interpretation along these lines: negative pointwise information 
\begin{equation}
i(s;t) := \log \frac{p(t|s)}{p(t)}
\end{equation}
means that an observer who has observed $S=s$ is less likely to guess the correct target value $t$ when conditioning on this observation (using the posterior distribution of $T$ given $S=s$) than when relying solely on their prior knowledge $p(t)$. In this case, the observation actively misleads.

Now, for mutual information this cannot be the case \textit{on average} over all observations. But the question is whether this might still reasonably happen for PID quantities. One possible line of argument can be developed through the $I_\cap^{sx}$ measure \cite{makkeh2021introducing}. As mentioned above $I_\cap^{sx}$ conceptualizes redundancy in terms of shared knowledge~\cite{gutknecht2023information}, in particular in terms of the pointwise mutual information this shared knowledge provides about the target realization. It is simple to show that relying on this shared knowledge to guess target realizations can indeed be misleading in the same sense explained above, even on average. Furthermore, $I_\cap^{sx}$ can be decomposed into informative and misinformative parts~\cite{makkeh2021introducing, finn2018probability}, each of which is non-negative, as is the resulting atom-level decomposition. In this way the negativity is completely explained in terms of this possibility of misinformation. This suggests that negative atoms might be acceptable, but the burden of proof lies with any PID proposal that allows them: one must demonstrate that negative values admit a coherent interpretation, for instance as misinformation in the sense just described. 

For the inconsistency involving the target chain rule, however, allowing negative atoms is not enough. By Proposition~\ref{prop:nr_fr}, a PID satisfying the target chain rule and re-encoding invariance must give up one of two much more elementary requirements: that a source without information about the target shares none of it with another source (\textit{no information, no redundancy}), or that two sources which each determine the target share all of it (\textit{full information, full redundancy}). The measures in Table~\ref{tab:measures_comparison} that satisfy the target chain rule and re-encoding invariance, $I_\cap^{sx}$ and $R_{min}$, give up the first. In the XOR gate both sources are independent of the target, yet $I_\cap^{sx}$ assigns them a redundancy of $-0.585$ bits~\cite{makkeh2021introducing}, and if the target is a copy of $S_1$ and $S_2$ is independent of it, $R_{min}$ assigns one bit of redundancy~\cite{finn2018pointwise}.

The shared-knowledge reading of $I_\cap^{sx}$~\cite{gutknecht2023information} shows how this can make sense. In the XOR gate, $S_1$ and $S_2$ are independent fair bits and $T=S_1\oplus S_2$. Suppose that $S_1=0$ and $S_2=0$, so that the actual target value is $T=0$, and consider two observers, one seeing $S_1=0$ and the other $S_2=0$. Both know that $S_1=0$ or $S_2=0$, since this follows from either observation: this is their shared knowledge. It rules out only the outcome $S_1=S_2=1$ and leaves three equally likely outcomes, $(0,0)$, $(0,1)$, and $(1,0)$, of which only the first gives $T=0$. The shared knowledge therefore lowers the probability of the actual target value from $1/2$ to $1/3$: it is misleading. Each observation on its own is nonetheless uninformative, since $T=0$ still has probability $1/2$ given $S_1=0$ alone, or $S_2=0$ alone. What each observation adds beyond the shared knowledge thus raises the probability back from $1/3$ to $1/2$. $I_\cap^{sx}$ counts the shared part as redundancy and the rest as unique information, which therefore have opposite signs. By symmetry, the same holds for every outcome, which gives the redundancy of $\log_2(2/3)\approx-0.585$ bits mentioned above. Finn and Lizier give a related reading of $R_{min}$ in terms of overlapping probability-mass exclusions~\cite{finn2018pointwise}. Giving up ``no information, no redundancy'' thus has a principled motivation, but, as for negative atoms, the burden of proof lies with any measure that does so.

\paragraph{Rejecting the consistency equations?} Instead of rejecting one of the properties \textbf{ID}, \textbf{REI}, \textbf{TCR}, or \textbf{LP} one might consider a more radical response: changing the very definition of PID by rejecting the consistency equations \eqref{eq:consistency}. The approach to PID suggested by Kolchinsky \cite{Kolchinsky2022} does precisely this, by allowing the information atoms not to sum to the joint mutual information. Lyu et al.~\cite{lyu2025whole} likewise replace the consistency equations by a modified summation rule (see also~\cite{lyu2026multivariate}). While this formally sidesteps the inconsistencies we have established, it creates a fundamental problem: the result is no longer a \emph{decomposition} of mutual information. 

More generally, violations of the consistency equations undermine the \textit{interpretability} of PID components. The consistency equations embody the minimal requirements for interpreting the components as measures of unique, redundant, and synergistic information. Consider the case $n=2$, where the consistency equations demand that redundancy and unique information from source $S_1$ sum to the mutual information $I(S_1;T)$:
\begin{equation}
\Pi(\{1\}\{2\}) + \Pi(\{1\}) = I(S_1;T).
\end{equation}
This relationship is necessary if these components are to deserve their names. They partition the mutual information in a \textit{logically exclusive and exhaustive way}: information in $S_1$ about $T$ can be either also contained in $S_2$ (redundancy) or not (unique to $S_1$). Together, these components should therefore account for the full mutual information $I(S_1;T)$. Accordingly, if this equation is violated, we cannot meaningfully interpret $\Pi(\{1\}\{2\})$ as redundancy or $\Pi(\{1\})$ as unique information (again, if we could, they would have to add up). 

The same logic applies to all consistency equations: the information atoms $\Pi(f)$ with $f(\mathbf{a})=1$ are intended to partition the mutual information term $I(\mathbf{a};T)$ according to which source collections do or do not provide that information (as specified by $f$). So together they must account for the full $I(\mathbf{a};T)$.  Therefore, if the consistency equations are violated, at least one atom cannot have its intended interpretation (if they all did, they should add up). For these reasons, abandoning the consistency equations comes at a substantive cost: the result may not be a decomposition of mutual information at all, and we lose the interpretability in terms of unique, redundant, and synergistic information, or more generally as information standing in specific parthood relationships to mutual information terms.

\section{Conclusion}
In this work, we have identified a fundamental incompatibility that arises in Partial Information Decomposition between several principles widely regarded as foundational within classical information theory: non-negativity of information, its invariance under re-encoding, and the existence of a chain rule capturing the idea that information about a multivariate target can be decomposed into information about its individual components. Our results show that these principles cannot, in general, be translated simultaneously into the domain of PID. 

Our main result holds in a considerably stronger form. Invariance under re-encoding is needed only for swapping the components of a composite target, and non-negativity only at two extremes: a source without information about the target shares none of it with another source, and two sources that each determine the target share all of it. A PID with a target chain rule that is invariant under swapping target components must therefore give up one of these elementary requirements.

We also show that a closely related inconsistency persists for three or more sources when the chain rule is replaced by the PID-specific (independent) identity property, strengthening earlier results in the literature. Overall, our results help to identify which disagreements in the PID literature reflect genuine conceptual trade-offs rather than incomplete constructions, thereby bringing clarity to a number of long-standing debates. 

Against this background, one possible reaction is to declare the PID programme untenable: if the fundamental properties of classical information theory discussed here cannot be jointly preserved, then perhaps one should abandon PID altogether. This is certainly an understandable position, and in that case researchers retain a rich toolbox of alternative multivariate measures—such as co-information or the redundancy–synergy index—that capture aspects of redundancy and synergy without aiming for a full decomposition and which can be computed entirely within classical information theory. Nevertheless, we believe that this conclusion would be premature and ultimately too strong.

First, as discussed in the previous section, the assumptions entering our inconsistency results need not be regarded as non-negotiable. In particular, there are principled reasons to relax or reject the target chain rule, local positivity, or the identity property. 

Second, PID has already demonstrated considerable practical value, especially in machine learning \cite{shwartz2024compress, makkeh2025general, wollstadt2023rigorous, liang2023quantifying, schneider2025neuron, dewan2024diffusion}. For instance, recent work on \textit{infomorphic neural networks} \cite{makkeh2025general} shows that PID quantities can be used to define local learning rules, where neurons optimise PID-based objective functions (e.g., promoting redundancy). The resulting networks learn effectively across supervised, unsupervised, and memory tasks. The success of these models shows that PID quantities can function operationally even if their interpretation as valid measures of unique, redundant, or synergistic components is contested, for example due to violations of local positivity.

\ack{We thank Alberto Liardi, Keenan Down, and Lionel Barnett for helpful discussions related to this work. This work was supported by the DFG Priority Programme SPP 2205 “Evolutionary Optimization of Neuronal Processing”.}

\section{Appendix}

\subsection{Proof of Proposition \ref{prop:equiv_id}} \label{app:equiv_id}
\begin{proof}
We establish the equivalences by showing (i) $\Rightarrow$ (ii) $\Rightarrow$ (iii) $\Rightarrow$ (iv) $\Rightarrow$ (i), using the fundamental relationships between different PID-inducing concepts established in \cite{gutknecht2025babel}.

\noindent\textbf{(i) $\Rightarrow$ (ii):}
Assume \textbf{ID} holds, i.e., $I_\cap(\{1\},\{2\};(S_1,S_2)) = I(S_1;S_2)$.

By the inclusion-exclusion principle for redundancy and union information (Proposition 2 in the electronic supplementary material of \cite{gutknecht2025babel}), we have
\begin{equation}
I_\cap(\{1\},\{2\};T) + I_\cup(\{1\},\{2\};T) = I(S_1;T) + I(S_2;T).
\end{equation}

For $T=(S_1,S_2)$, we have $I(S_1;T) = H(S_1)$ and $I(S_2;T) = H(S_2)$. Therefore:
\begin{align}
I_\cup(\{1\},\{2\};T) &= H(S_1) + H(S_2) - I_\cap(\{1\},\{2\};T) \\
&= H(S_1) + H(S_2) - I(S_1;S_2) \\
&= H(S_1,S_2).
\end{align}

\noindent\textbf{(ii) $\Rightarrow$ (iii):}
Assume $I_\cup(\{1\},\{2\};T) = H(S_1,S_2)$.

By the complementation principle for union information and weak synergy (Proposition 3 in the electronic supplementary material of \cite{gutknecht2025babel}), we have
\begin{equation}
I_\cup(\{1\},\{2\};T) + I_{ws}(\{1\},\{2\};T) = I(S_1,\ldots,S_n;T).
\end{equation}

For $T=(S_1,S_2)$, we have $I(S_1,\ldots,S_n;T) = H(S_1,S_2)$, since $T$ is a function of $S_1$ and $S_2$. Therefore:
\begin{align}
I_{ws}(\{1\},\{2\};T) &= H(S_1,S_2) - I_\cup(\{1\},\{2\};T) \\
&= H(S_1,S_2) - H(S_1,S_2) \\
&= 0.
\end{align}

\noindent\textbf{(iii) $\Rightarrow$ (iv):}
Assume $I_{ws}(\{1\},\{2\};T) = 0$.

By the inclusion-exclusion principle for weak synergy and vulnerable information (Proposition 2 in the electronic supplementary material of \cite{gutknecht2025babel}), we have
\begin{equation}
I_{ws}(\{1\},\{2\};T) + I_{vul}(\{1\},\{2\};T) = I_{ws}(\{1\};T) + I_{ws}(\{2\};T).
\end{equation}

For $T=(S_1,S_2)$, by the self-synergy property we have $I_{ws}(\{1\};T) = I(S_2,\ldots,S_n;T|S_1) = H(S_2|S_1)$ and $I_{ws}(\{2\};T) = I(S_1,S_3,\ldots,S_n;T|S_2) = H(S_1|S_2)$, again since $T$ is a function of $S_1$ and $S_2$. Therefore:
\begin{align}
I_{vul}(\{1\},\{2\};T) &= H(S_2|S_1) + H(S_1|S_2) - I_{ws}(\{1\},\{2\};T) \\
&= H(S_2|S_1) + H(S_1|S_2) - 0 \\
&= H(S_1|S_2) + H(S_2|S_1).
\end{align}
\noindent\textbf{(iv) $\Rightarrow$ (i):}
Assume $I_{vul}(\{1\},\{2\};T) = H(S_1|S_2) + H(S_2|S_1)$.

By the complementation principle for vulnerable and redundant information (Proposition 3 in the electronic supplementary material of \cite{gutknecht2025babel}), we have
\begin{equation}
I_{vul}(\{1\},\{2\};T) + I_\cap(\{1\},\{2\};T) = I(S_1,\ldots,S_n;T).
\end{equation}

For $T=(S_1,S_2)$, we have $I(S_1,\ldots,S_n;T) = H(S_1,S_2)$, since $T$ is a function of $S_1$ and $S_2$. Therefore:
\begin{align}
I_\cap(\{1\},\{2\};T) &= H(S_1,S_2) - I_{vul}(\{1\},\{2\};T) \\
&= H(S_1,S_2) - (H(S_1|S_2) + H(S_2|S_1)) \\
&= H(S_2) -  H(S_2|S_1) \\
&= I(S_1;S_2).
\end{align}
\end{proof}

\subsection{Proof of Proposition \ref{prop:equiv_tcr}}\label{app:equiv_tcr}

\begin{proof}
For a condition $\mathcal{C}$, write $I^{\mathcal{C}}(\gamma;T)=\sum_{f:\,\mathcal{C}(\gamma;f)}\Pi(f;T)$ for the measure of $\mathcal{C}$-information associated with the atoms, where $\gamma=(\mathbf{a}_1,\ldots,\mathbf{a}_m)$ ranges over tuples of source collections. Conditional atoms and conditional $\mathcal{C}$-information are both obtained by evaluating the same functionals under the conditional distribution and averaging, so the same relation holds between $\Pi(f;T_2\mid T_1)$ and $I^{\mathcal{C}}(\gamma;T_2\mid T_1)$.

\noindent\emph{(i) $\Rightarrow$ (ii).} Let $\mathcal{C}$ be PID-inducing and assume \textbf{C-TCR}. Since the linear system relating $I^{\mathcal{C}}$ to the atoms has exactly one solution, the atoms are fixed linear combinations of finitely many values of $I^{\mathcal{C}}$: there are coefficients $K(f;\gamma)$, depending only on $\mathcal{C}$, such that
\begin{equation}
\Pi(f;T) \;=\; \sum_{\gamma} K(f;\gamma)\, I^{\mathcal{C}}(\gamma;T) \quad \text{for all } f \in \mathcal{B}_n.
\end{equation}
The same coefficients apply under every conditional distribution, and hence to $\Pi(f;T_2\mid T_1)$ and $I^{\mathcal{C}}(\gamma;T_2\mid T_1)$. Applying \textbf{C-TCR} term by term therefore gives
\begin{equation}
\Pi(f;T) \;=\; \sum_{\gamma} K(f;\gamma)\, I^{\mathcal{C}}(\gamma;T_1) + \sum_{\gamma} K(f;\gamma)\, I^{\mathcal{C}}(\gamma;T_2\mid T_1) \;=\; \Pi(f;T_1) + \Pi(f;T_2\mid T_1).
\end{equation}

\noindent\emph{(ii) $\Rightarrow$ (iii).} For every condition $\mathcal{C}$, summing \textbf{A-TCR} over all $f$ with $\mathcal{C}(\gamma;f)$ gives
\begin{equation}
I^{\mathcal{C}}(\gamma;T) \;=\; I^{\mathcal{C}}(\gamma;T_1) + I^{\mathcal{C}}(\gamma;T_2\mid T_1).
\end{equation}

\noindent\emph{(iii) $\Rightarrow$ (i).} Redundant information is PID-inducing (Section~\ref{sec:background:pid_inducing}), so $\mathcal{C}=\mathcal{C}_{\mathrm{red}}$ is an instance of (i).
\end{proof}

\subsection{Proof of Lemma \ref{lem:LP_pw_red}} \label{app:ext_xorsc_lemma}

\begin{proof}
We have $I(S_1,\ldots,S_n;T) = I(S_1,S_2,S_3;T) = 2$, and hence
\begin{equation}
R_3 := \sum_{i=1}^{3} I(S_i;T) - I(S_1,\ldots,S_n;T) = 3 - 2 = 1.
\end{equation}
Expanding all terms via the consistency equations~\eqref{eq:consistency},
each atom $\Pi(\alpha;T)$, $\alpha \in \mathcal{A}_n$, appears in the first
sum with multiplicity
\begin{equation}
r_3(\alpha) := \bigl|\{\, i \in \{1,2,3\} \mid \{i\} \in \alpha \,\}\bigr|,
\end{equation}
i.e. the degree of redundancy restricted to the first three sources, and in
the second sum exactly once. Therefore,
\begin{equation}
R_3 = \sum_{\alpha:\,r_3(\alpha) \geq 2} (r_3(\alpha)-1)\,\Pi(\alpha;T)
    \;-\; \sum_{\alpha:\,r_3(\alpha) = 0} \Pi(\alpha;T) \;=\; 1,
\end{equation}
in complete analogy to Equation~\eqref{eq:rsi_degr_red_2}, to which this
reduces for $n=3$ (where $r_3 = r$). Under \textbf{LP}, at least one atom
with $r_3(\alpha) \geq 2$ must be strictly positive. Any such $\alpha$
contains singletons $\{i\}$ and $\{j\}$ for distinct $i,j \in \{1,2,3\}$, so
that $\Pi(\alpha;T)$ contributes to the pairwise redundancy
$I_\cap(\{i\},\{j\};T)$, which is hence strictly positive as well.
\end{proof}

\subsection{\texorpdfstring{A PID satisfying local positivity and the target chain rule}{A PID satisfying local positivity and the target chain rule}}\label{app:chained_pid}

We construct, for every $n\geq 2$, a PID satisfying \textbf{LP} and \textbf{TCR}, and hence also \textbf{ID}, as needed for Propositions~\ref{prop:minimality} and~\ref{prop:minimality_id}. By Theorem~\ref{thm:lp_tcr_incompatible}, such a PID must violate \textbf{REI}.

\emph{Construction.} Fix a base PID $\{\Pi^{b}(f;T)\}_{f\in\mathcal{B}_n}$ satisfying \textbf{LP}, for instance the one induced by $I_{\min}$~\cite{williams2010nonnegative}. For a target given as a random vector $T=(T_1,\ldots,T_k)$, with $T_{<l}:=(T_1,\ldots,T_{l-1})$, define the \emph{chained PID}
\begin{equation}\label{eq:chained_pid}
\Pi^{\mathrm{ch}}(f;T) := \sum_{l=1}^{k} \Pi^{b}(f;T_l\mid T_{<l}),
\end{equation}
where the conditional atoms are those of Definition~\ref{def:cond_atoms} and the term $l=1$ is $\Pi^{b}(f;T_1)$. As usual, a vector formed from random vectors is their concatenation; for instance, $((T_1,T_2),T_3)=(T_1,T_2,T_3)$. On a target with a single component, $\Pi^{\mathrm{ch}}$ is the base PID; on a composite target, it obeys the chain rule by construction. The associated redundancy $I_\cap^{\mathrm{ch}}(\alpha;T):=\sum_{\beta\preceq\alpha}\Pi^{\mathrm{ch}}(\beta;T)$ is, by linearity, the chained base redundancy, $I_\cap^{\mathrm{ch}}(\alpha;T)=\sum_{l} I_\cap^{b}(\alpha;T_l\mid T_{<l})$.

\begin{prop}\label{prop:chained_pid}
The chained PID is a PID, i.e.\ a family of functionals of the joint distribution of $(\mathbf{S},T_1,\ldots,T_k)$ satisfying the consistency equations~\eqref{eq:consistency}, and it satisfies \textbf{LP}, \textbf{TCR}, and \textbf{ID}. It violates \textbf{REI}; in fact, it is not even invariant under permutations of the target components.
\end{prop}
\begin{proof}
Write $\Pi^{b,t_{<l}}$ for the base PID evaluated under the conditional distribution of $(\mathbf{S},T_l)$ given $T_{<l}=t_{<l}$, so that $\Pi^{b}(f;T_l\mid T_{<l})=\sum_{t_{<l}}p(t_{<l})\,\Pi^{b,t_{<l}}(f;T_l)$. Each such term depends only on $p(t_{<l})$ and $p(\mathbf{s},t_l\mid t_{<l})$, and hence on the joint distribution. Since the base PID satisfies the consistency equations for every joint distribution, it satisfies them in particular under each of these conditional distributions, with $I(\mathbf{a};T_l\mid T_{<l}=t_{<l})$ on the right-hand side. Averaging over $t_{<l}$ yields the conditional consistency equations
\begin{equation}\label{eq:cond_consistency}
\sum_{\substack{f\in\mathcal{B}_n\\ f(\mathbf{a})=1}}\Pi^{b}(f;T_l\mid T_{<l}) = I(\mathbf{a};T_l\mid T_{<l}) \quad \text{for all } \mathbf{a}\subseteq[n].
\end{equation}

\emph{Consistency.} For every $\mathbf{a}\subseteq[n]$, inserting~\eqref{eq:chained_pid}, exchanging the sums, and using~\eqref{eq:cond_consistency} and the chain rule of mutual information gives
\begin{equation}
\begin{aligned}
\sum_{\substack{f\in\mathcal{B}_n\\ f(\mathbf{a})=1}}\Pi^{\mathrm{ch}}(f;T)
&=\sum_{\substack{f\in\mathcal{B}_n\\ f(\mathbf{a})=1}}\ \sum_{l=1}^{k}\Pi^{b}(f;T_l\mid T_{<l})
=\sum_{l=1}^{k}\ \sum_{\substack{f\in\mathcal{B}_n\\ f(\mathbf{a})=1}}\Pi^{b}(f;T_l\mid T_{<l})\\
&=\sum_{l=1}^{k}I(\mathbf{a};T_l\mid T_{<l})
=I(\mathbf{a};T).
\end{aligned}
\end{equation}

\emph{Local positivity.} By \eqref{eq:chained_pid}, each atom $\Pi^{\mathrm{ch}}(f;T)$ is a combination, with non-negative weights $p(t_{<l})$, of atoms of the base PID under conditional distributions, which are non-negative by \textbf{LP} of the base PID.

\emph{Target chain rule.} Let $T=(T_1,T_2)$, where $T_1$ and $T_2$ may themselves be random vectors. In~\eqref{eq:chained_pid}, the terms for the components of $T_1$ sum to $\Pi^{\mathrm{ch}}(f;T_1)$. Each remaining term conditions a component of $T_2$ on $T_1$ and on the preceding components of $T_2$. Since conditioning on $T_1$ and then on the preceding components of $T_2$ is the same as conditioning on both at once, these terms sum to $\Pi^{\mathrm{ch}}(f;T_2\mid T_1)$, the chained PID of $T_2$ under the conditional distributions given $T_1$, averaged over $T_1$ (Definition~\ref{def:cond_atoms}). Hence
\begin{equation}
\Pi^{\mathrm{ch}}(f;T)=\Pi^{\mathrm{ch}}(f;T_1)+\Pi^{\mathrm{ch}}(f;T_2\mid T_1),
\end{equation}
i.e.\ the chained PID satisfies \textbf{A-TCR}. By Proposition~\ref{prop:equiv_tcr}, the associated measure of redundant information $I_\cap^{\mathrm{ch}}$ therefore satisfies \textbf{TCR}.

\emph{Identity property.} This follows from \textbf{LP} and \textbf{TCR} by Lemma~\ref{lem:lp_tcr_id}.

\emph{Violation of \textbf{REI}.} Since the chained PID satisfies \textbf{LP} and \textbf{TCR}, it violates \textbf{REI} by Theorem~\ref{thm:lp_tcr_incompatible}, and it is not invariant under permutations of the target components by Proposition~\ref{prop:perm_inv}. The following example shows this explicitly for $n=2$ and the base PID induced by $I_{\min}$.
\end{proof}

\emph{Example.} Let $n=2$, let $S_1,S_2$ be independent fair bits, let $X:=S_1\oplus S_2$ as in the proof of Theorem~\ref{thm:lp_tcr_incompatible}, and let the base PID be induced by $I_{\min}$. We write the atoms as $(\Pi(\{1\}\{2\}),\Pi(\{1\}),\Pi(\{2\}),\Pi(\{1,2\}))$, i.e.\ redundancy, the two unique informations, and synergy. Two facts determine the redundancies: a source that is independent of the target (in particular, a constant source) has zero specific information about every target value, and a source that determines a uniform binary target has one bit of specific information about every target value. The remaining atoms then follow from the consistency equations. This gives the base PIDs
\begin{center}
\begin{tabular}{lll}
target & reason & base atoms\\\hline
$S_1$ & $S_2$ is independent of $S_1$ & $(0,1,0,0)$\\
$X$ given $S_1=s_1$ & $S_1$ is constant, $S_2$ determines $X$ & $(0,0,1,0)$\\
$X$ & each source alone is independent of $X$ & $(0,0,0,1)$\\
$S_1$ given $X=x$ & both sources determine $S_1$, as $S_2=S_1\oplus x$ & $(1,0,0,0)$
\end{tabular}
\end{center}
and hence
\begin{equation}
\Pi^{\mathrm{ch}}(\,\cdot\,;(S_1,X)) = (0,1,1,0), \qquad \Pi^{\mathrm{ch}}(\,\cdot\,;(X,S_1)) = (1,0,0,1).
\end{equation}
Both assignments are non-negative and consistent, since $I(S_1;T)=I(S_2;T)=1$ and $I(S_1,S_2;T)=2$ for both targets. The two targets differ only by the order of their components, so \textbf{REI} would force equal atoms. The chained PID also violates \textbf{REI} for re-encodings that are not permutations: by the same reasoning, the two-bit copy $(S_1,S_2)$ receives $(0,1,1,0)$, whereas its bijective re-encoding $(X,S_2)$ receives $(1,0,0,1)$. Plain $I_{\min}$, by contrast, assigns $(1,0,0,1)$ to all of these targets. It satisfies \textbf{LP} and \textbf{REI} but violates \textbf{TCR}, since chaining the redundancy at $(S_1,X)$ gives $0+0\neq 1$.

\bibliographystyle{RS} %%%% .BST file

\bibliography{sample} %%%%% .Bib file

\end{document}